\documentclass[10pt,journal,compsoc]{IEEEtran}

\ifCLASSOPTIONcompsoc
  \usepackage[nocompress]{cite}
\else
  \usepackage{cite}
\fi
\ifCLASSINFOpdf
   \usepackage[pdftex]{graphicx}
   \graphicspath{{./figures/}}
\else
\fi
\usepackage{amsmath}

\usepackage{color}
\definecolor{darkgray}{rgb}{0.66, 0.66, 0.66}

\allowdisplaybreaks

\ifCLASSOPTIONcompsoc
  \usepackage[caption=false,font=footnotesize,labelfont=sf,textfont=sf]{subfig}
\else
  \usepackage[caption=false,font=footnotesize]{subfig}
\fi
\usepackage{url}
\usepackage{amsfonts}
\usepackage{amssymb}
\usepackage{amsthm}
\usepackage{commath}
\usepackage{amsmath}

\DeclareMathOperator*{\argmin}{arg\,min}
\newtheorem{theorem}{Theorem}

\newtheorem{lemma}[theorem]{Lemma}
\newtheorem{proposition}[theorem]{Proposition}

\usepackage{algorithm} 
\usepackage{algorithmic}
\usepackage{booktabs}
\usepackage{multirow}
\usepackage{graphicx}
\usepackage{bbold} 
\usepackage{cleveref}
\usepackage[table,xcdraw]{xcolor}
\usepackage[normalem]{ulem}
\useunder{\uline}{\ul}{}

\newcommand*\mystrut[1]{\vrule width0pt height0pt depth#1\relax}

\begin{document}
%
\title{Multiple Kernel Representation Learning on Networks}
%
%
%
%

\author{Abdulkadir~Çelikkanat, Yanning Shen, and Fragkiskos~D.~Malliaros
\IEEEcompsocitemizethanks{\IEEEcompsocthanksitem  A. Çelikkanat is with the Department of Applied Mathematics and Computer Science, Technical University of Denmark, 2800 Kgs. Lyngby, Denmark. The research has been mostly conducted while the author was with Paris-Saclay University, CentraleSupélec, Inria, Centre for Visual Computing (CVN), 91190 Gif-Sur-Yvette, France \protect\\
E-mail: abdcelikkanat@gmail.com
\IEEEcompsocthanksitem  Y. Shen is with the Department of EECS and the Center for Pervasive Communications and Computing, University of California, Irvine, CA 92697, USA \protect\\
E-mail: yannings@uci.edu
\IEEEcompsocthanksitem  F. D. Malliaros is with Paris-Saclay University, CentraleSupélec, Inria, Centre for Visual Computing (CVN), 91190 Gif-Sur-Yvette, France \protect\\
E-mail: fragkiskos.malliaros@centralesupelec.fr
}
\thanks{Manuscript received XXX; revised XXX.}}
\IEEEtitleabstractindextext{%
\begin{abstract}
Learning representations of nodes in a low dimensional space is a crucial task with numerous interesting applications in network analysis, including link prediction, node classification, and visualization. Two popular approaches for this problem are \textit{matrix factorization} and \textit{random walk}-based models. In this paper, we aim to bring together the best of both worlds, towards learning node representations. In particular, we propose a weighted matrix factorization model that encodes random walk-based information about nodes of the network. The benefit of this novel formulation is that it enables us to utilize kernel functions without realizing the exact proximity matrix so that it enhances the expressiveness of existing matrix decomposition methods with kernels and alleviates their computational complexities. We extend the approach with a multiple kernel learning formulation that provides the flexibility of learning the kernel as the linear combination of a dictionary of kernels in data-driven fashion. We perform an empirical evaluation on real-world networks, showing that the proposed model outperforms baseline node embedding algorithms in downstream machine learning tasks.
\end{abstract}

}

\maketitle

\IEEEdisplaynontitleabstractindextext

%
\IEEEpeerreviewmaketitle


\ifCLASSOPTIONcompsoc
\IEEEraisesectionheading{\section{Introduction}\label{sec:introduction}}
\else
\section{Introduction}
\label{sec:introduction}
\fi
\IEEEPARstart{W}{ith} the development in data production, storage and consumption, graph data is becoming omnipresent; data from diverse disciplines can be represented as graph structures with prominent examples including various social, information, technological, and biological networks. Developing machine learning algorithms to analyze, predict, and make sense of graph data has become a crucial task with a plethora of cross-disciplinary applications \cite{survey_hamilton_leskovec, GRL-survey-ieeebigdata20}. The major challenge in machine learning on graph data concerns the encoding of information about the graph structural properties into the learning model. To this direction,  \textit{network representation learning} (NRL), a recent paradigm in network analysis, has received substantial attention thanks to its outstanding performance in downstream tasks, including link prediction and classification. Representation learning techniques mainly target to embed the nodes of the graph into a lower-dimensional space in such a way that desired relational properties among graph nodes are captured by the similarity of the representations in the embedding space \cite{deepwalk, node2vec, laplacian_eigenmap, grarep, hope, adaptive_diffusion}. 

\par The area of NRL has been highly impacted by traditional nonlinear dimensionality reduction approaches \cite{laplacian_eigenmap, locally_linear_embedding}. Specifically, many proposed models had initially concentrated on learning node embeddings relying on matrix factorization techniques that encode structural node similarity\cite{laplacian_eigenmap, grarep, hope}. Nevertheless, most of those approaches are not very efficient for large scale networks, mainly due to the high computational and memory cost required to perform matrix factorization. Besides, most such models require the exact realization of the target matrix \cite{survey_hamilton_leskovec, GRL-survey-ieeebigdata20}.

\par Inspired by the field of natural language processing \cite{word2vec},  random-walk based embedding has gained considerable attention (e.g., \cite{deepwalk,node2vec,biasedwalk, tne, efge-aaai20, InfiniteWalk-kdd20}). Typically, these methods firstly produce a set of node sequences by following certain random walk strategy. Then, node representations are learned by optimizing the relative co-occurrence probabilities of nodes within the random walks. Although such random walk models follow similar approaches in modeling the relationships between nodes and optimizing the embedding vectors, their difference mainly stems from the way in which node sequences are sampled \cite{deepwalk,node2vec}.

\par On a separate topic, \textit{kernel functions} have often been introduced along with popular learning algorithms, such as PCA \cite{kernel_pca}, SVM \cite{svm}, Spectral Clustering \cite{kernel_clustering}, and Collaborative Filtering \cite{kernel_matrix_fact}, to name a few. Most of traditional learning models are insufficient to capture the underlying substructures of complex datasets, since they rely  on linear techniques to model  nonlinear patterns existing in data. Kernel functions \cite{kernelsML08}, on the other hand, allow mapping non-linearly separable points into a (generally) higher dimensional feature space, so that the inner product in the new space can be computed without needing to compute the exact feature maps---bringing further computational benefits. Besides, to further reduce model bias, \textit{multiple kernel learning} approaches have been proposed to learn optimal combinations of kernel functions \cite{mkl-jmlr11}. Nevertheless, despite their wide applications in various fields \cite{kernelsML08, Wang2012rmSM}, kernel and multiple kernel methods have not been thoroughly investigated for learning node embeddings.

\par In this paper, we aim at combining matrix factorization and random walks in a kernelized model for learning node embeddings. The potential advantage of such a modeling approach is that it allows for leveraging and combining the elegant mathematical formulation offered by matrix factorization with the expressive power of random walks to capture a notion of ``stochastic'' node similarity in an efficient way. More importantly, this formulation enables leveraging kernel functions in the node representation learning task. Because of the nature of matrix factorization-based models, node similarities can be viewed as inner products of vectors lying in a latent space---which allows to utilize kernels towards interpreting the embeddings in a higher dimensional feature space using non-linear maps. Besides, multiple kernels can be utilized to learn more discriminative node embeddings. Note that, although graph kernels is a well-studied topic in graph learning \cite{KriegeJM20}, it mainly focuses on graph classification---a task outside of the scope of this paper. To the best of our knowledge, random walk-based multiple kernel matrix factorization has not been studied before for learning node embeddings.

\par The main contributions of the present work can be summarized as follows:
\begin{itemize}
    \item We propose \textsc{KernelNE} (Kernel Node Embeddings), a novel approach for learning node embeddings by incorporating kernel functions with models relying on weighted matrix factorization, encoding random walk-based structural information of the graph. We further examine the performance of the model with different types of kernel functions.
    \item To further improve expressiveness, we introduce \textsc{MKernelNE}, a multiple kernel learning formulation of the model. It extends the kernelized weighted matrix factorization framework by learning a linear combination of a predefined set of kernels. 
    \item We demonstrate how both models (single and multiple kernel learning) leverage negative sampling to efficiently compute node embeddings. 
    \item We extensively evaluate the proposed method's performance in node classification and link prediction. We show that the proposed models generally outperform well-known baseline methods on various real-world graph datasets. Besides, due to the efficient model optimization mechanism, the running time is comparable to the one of random walk models.
\end{itemize}

\noindent \textbf{Notation.} We use the notation $\mathbf{M}$ to denote a matrix, and the term $\mathbf{M}_{(v,u)}$ represents the entry at the $v$'th row and $u$'th column of the matrix $\mathbf{M}$. $\mathbf{M}_{(v,:)}$ indicates the $v$'th row of the matrix.

 \vspace{.2cm}

\noindent \textbf{Source code.} The C++ implementation of the proposed methodology can be reached at: \texttt{\url{https://abdcelikkanat.github.io/projects/kernelNE/}}.

\section{Related Work} \label{sec:related-work}

\noindent \textbf{Node embeddings.} Traditional dimension reduction methods such as Principal Component Analysis (PCA) \cite{pca}, Locally Linear Embeddings \cite{locally_linear_embedding}, and ISOMAP \cite{isomap} that rely on the matrix factorization techniques,  have considerably influenced early approaches on learning node embeddings. For instance, \textsc{HOPE} \cite{hope} learns embedding vectors by factorizing a higher-order proximity matrix using SVD, while \textsc{M-NMF} \cite{nmf} leverages non negative matrix factorization aiming to preserve the underlying community structure of the graph. Similarly, \textsc{NetMF} \cite{netmf} and its scalable variant \textsc{NetSMF} \cite{netsmf-www2019}, learn embeddings by properly factorizing a target matrix that has been designed based on the Positive Pointwise Mutual Information (PPMI) of random walk proximities. Similarly, \textsc{SDAE} \cite{SDAE} uses the PPMI matrix with stacked denoising autoencoders to learn the representations. SDNE \cite{sdne} adapts a deep neural network architecture to capture the highly non-linear network patterns. Nevertheless, as we have mentioned in Sec. \ref{sec:introduction}, most  of these models suffer from high time complexity, while at the same time, they assume an exact realization of the target matrix. For a detailed description of matrix factorization models for node embeddings, the reader can refer to \cite{GRL-survey-ieeebigdata20,survey_hamilton_leskovec, hamilton-grl20}.

\par To further improve the expressiveness as well as to alleviate the computational burden of matrix factorization approaches, models that rely on random walk sampling have been developed (e.g., \cite{deepwalk,node2vec,biasedwalk, tne, tne_journal, efge-aaai20, InfiniteWalk-kdd20, watchyourstep, verse-www18}). These models aim at modeling \textit{center-context} node relationships in random walk sequences, leveraging advances on learning word embeddings in natural language processing \cite{word2vec}. Due to the flexible way that random walks can be defined, various formulations have been proposed, with the two most well-known being \textsc{DeepWalk} \cite{deepwalk} and \textsc{Node2Vec} \cite{node2vec}. Most of these models though, are based on shallow architectures which do not allow to model complex non-linear data relationships, limiting the prediction capabilities on downstream tasks (e.g., node classification) \cite{hamilton-grl20}. There are also recent efforts towards developing scalable models relying on fast random projections for very large networks \cite{prone, nodesketch, randne, fastrp}.

\par As will be presented shortly, in this paper we aim to bring together the best of both worlds, by proposing an efficient random walk-based matrix factorization framework that allows to learn informative embeddings.

\vspace{.2cm}

\noindent \textbf{Kernel methods.} Although most algorithms in the broader field of machine learning have been developed for the linear case, real-world data often requires nonlinear models capable of unveiling the underlying complex relationships towards improving the performance of downstream tasks. To that end, kernel functions allow computing the inner product among data points in a typically high-dimensional feature space, in which linear models could still be applied, without explicitly computing the feature maps \cite{kernelsML08}. Because of the generality of the inner product, kernel methods have been widely used in a plethora of models, including Support Vector Machine \cite{svm}, PCA \cite{kernel_pca}, spectral clustering \cite{kernel_clustering, kernel-spectral-clustering-mapi10}, collaborative filtering \cite{kernel_matrix_fact}, and non-negative matrix factorization for image processing tasks \cite{nmf_kernel}, to name a few.

\par Kernel functions have also been utilized in the field graph analysis. At the node level, diffusion kernels and their applications \cite{Kondor02diffusionkernels} constitute notable instances. At the graph level, \textit{graph kernels} \cite{vishwanathan2010graph}, such as the random walk and Weisfeiler-Lehman kernels \cite{KriegeJM20}, have mainly been utilized to measure the similarity between a pair of graphs for applications such as graph classification. Besides, they have also used for capturing the temporal changes in time-evolving networks \cite{Melnyk2020}. Recent approaches have also been proposed to leverage graph kernels for node classification, but in a supervised manner  \cite{rethinking_kernels}. Other related applications of kernel methods on graphs include topology inference \cite{topology-inference-tsp17, inference-dynamic17}, signal reconstruction \cite{recostruction-tsp17}, and anomaly detection \cite{dots-spm19}.

\par The expressiveness of kernel methods can further be enhanced using multiple kernel functions in a way that the best possible combination of a set of predefined kernels can be learned \cite{Bach-MKL, mkl-jmlr11}. Besides improving prediction performance, multiple kernels have also been used to combine information from distinct heterogeneous sources (please see \cite{mkl-jmlr11} and \cite{WANG20213} for a detailed presentation  of several multiple kernel learning algorithms and their applications). Examples include recent approaches   that leverage multi-kernel strategies for graph and image datasets \cite{SALIM2020103534, 9040639}.

\par Despite the widespread applications of kernel and multiple kernel learning methods, utilizing them for learning node embeddings via matrix factorization in an unsupervised way is a problem that has not been thoroughly investigated. Previous works that are close to our problem settings, which follow a  methodologically different factorization approach (e.g., leveraging nonnegative matrix factorization \cite{kernel_matrix_fact}), targeting different application domains (e.g., collaborative filtering \cite{mknmf}). The multiple kernel framework for graph-based dimensionality reduction  proposed by Lin et al. \cite{mkl_for_dim_red} is also close to our work. Nevertheless, their work focuses mainly on leveraging multiple kernel functions to fuse image descriptors in computer vision tasks properly.

\par In this paper we propose  novel unsupervised models for node embeddings, that implicitly perform weighted matrix decomposition in a higher-dimensional space through kernel functions. The target pairwise node proximity matrix is properly designed to leverage information from random walk sequences, and thus, our models do not require the exact realization of this matrix. Both single and multiple kernel learning formulations are studied. Emphasis is also put on optimization aspects to ensure that node embeddings are computed efficiently. 
\section{Modeling and Problem Formulation}\label{sec:problem_formulation}

Let $G=(\mathcal{V}, \mathcal{E})$ be a graph where $\mathcal{V}=\{1,\ldots,n\}$ and $\mathcal{E}\subseteq\mathcal{V}\times\mathcal{V}$ are the vertex and edge sets, respectively. Our goal is to find node representations in a latent space, preserving properties of the network. More formally, we define the general objective function  of our problem as a weighted matrix factorization \cite{weighted_low_rank}, as follows:

\begin{align}\label{eq:main_obj_func}
    \argmin_{\mathbf{A},\mathbf{B}}\underbrace{\left\Vert\mathbf{W}\odot(\mathbf{M} - 
    \mathbf{A}\mathbf{B}^{\top})\right\Vert_F^2}_{\text{Error  term}} \!+\!\! \underbrace{\frac{\lambda}{2}\left( \Vert \mathbf{A} \Vert_F^2 +  \Vert \mathbf{B} \Vert_F^2 \right)}_{\text{Regularization term, $\mathcal{R}(\mathbf{A}, \mathbf{B})$}},
\end{align}

\noindent where $\mathbf{M}\in\mathbb{R}^{n\times n}$ is the target matrix constructed based on the desired properties of a given network, which is used to learn node embeddings $\mathbf{A}, \mathbf{B}\in \mathbb{R}^{n \times d}$, and $\Vert\cdot\Vert_F$ denotes the \textit{Frobenius} norm. We will use $\mathcal{R}(\mathbf{A},\mathbf{B})$ to denote the regularization term of Eq. \eqref{eq:main_obj_func}. Each element $\mathbf{W}_{v,u}$ of the weight matrix $\mathbf{W} \in \mathbb{R}^{n \times n}$ captures the importance of the approximation error between nodes $v$ and $u$, and $\odot$ indicates the \textit{Hadamard} product. Depending on the desired graph properties that we are interested in encoding, there are many possible alternatives to choose matrix $\mathbf{M}$; such include the number of common neighbors between a pair of nodes, higher-order node proximity based on the \textit{Adamic-Adar} or \textit{Katz}  indices \cite{hope}, as well leveraging multi-hop information \cite{grarep}. Here, we will design $\mathbf{M}$ as a sparse binary matrix utilizing information of random walks over the network.

As we have mentioned above, random walk-based node embedding models (e.g., \cite{deepwalk, node2vec, biasedwalk, epasto-splitter, NguyenLRAKK18, node_embed_comm, InfiniteWalk-kdd20}) have received great attention because of their good prediction performance and efficiency on large scale-networks. Typically, those models generate a set of node sequences by simulating random walks; node representations are then learned by optimizing a model which defines the relationships between nodes and their \textit{contexts} within the walks. More formally, for a random walk  $\textbf{w}=(w_1,\ldots,w_{L})$, the \textit{context set} of \textit{center} node $w_l\in\mathcal{V}$ appearing at the position $l$ of the walk $\textbf{w}$ is defined by $\{ w_{l-\gamma},\ldots,w_{l-1},w_{l+1},\ldots,w_{l+\gamma}\}$,  where $\gamma$ is the \textit{window size} which is the furthest distance between the \textit{center} and \textit{context} nodes. 
The embedding vectors are then obtained by maximizing the likelihood of occurrences of nodes within the context of given center nodes. Here, we will also follow a similar random walk strategy, formulating the problem using a matrix factorization framework.

Let $\mathbf{M}_{(v,u)}$ be a binary value which equals to $1$ if  node $u$ appears in the context of $v$ in any walk. Also, let $\mathbf{F}_{(v,u)}$ be $2\cdot \gamma\cdot \#(v)$, where $\#(v)$ indicates the total number of occurrences of  node $v$ in the generated walks. Setting each term $\mathbf{W}_{(v,u)}$ as the square root of $\mathbf{F}_{(v,u)}$, the objective function in Eq. \eqref{eq:main_obj_func} can be expressed under a random walk-based formulation as follows:

\begin{align}
	&\argmin_{\mathbf{A},\mathbf{B}}\Big\Vert\mathbf{W}\odot \big(\mathbf{M} - \mathbf{A}\mathbf{B}^{\top} \big)\Big\Vert_{F}^2 + 
	\mathcal{R}(\mathbf{A}, \mathbf{B}) \nonumber\\
	=&\argmin_{\mathbf{A},\mathbf{B}}\Big\Vert\sqrt{\mathbf{F}}\odot \big(\mathbf{M} - \mathbf{A}\mathbf{B}^{\top} \big)\Big\Vert_{F}^2 + \mathcal{R}(\mathbf{A}, \mathbf{B})\nonumber \nonumber\\
	=&\argmin_{\mathbf{A},\mathbf{B}}\!\!\sum_{v \in \mathcal{V}}\sum_{u \in \mathcal{V}}\mathbf{F}_{(v,u)} \!\Big(\mathbf{M}_{(v,u)}\!-\! \big \langle \mathbf{A}[u],\mathbf{B}[v]\big \rangle \!\Big)^2  \!\!\!\!+\! \mathcal{R}(\mathbf{A}, \mathbf{B})\nonumber
	\end{align}\vspace{-0.4cm}
	\begin{align}
	=&\argmin_{\mathbf{A},\mathbf{B}}\!\!\sum_{v \in \mathcal{V}}\sum_{u \in \mathcal{V}}\!\!2\!\cdot\!\gamma\!\!\cdot\! \#(v) \!\Big(\mathbf{M}_{(v,u)}\!\!-\! \big \langle \mathbf{A}[u],\mathbf{B}[v]\big \rangle \!\Big)^2  \!\!\!\!\!+\! \mathcal{R}(\mathbf{A}, \mathbf{B})\nonumber\\
	=&\argmin_{\mathbf{A},\mathbf{B}}\!\!\sum_{v \in \mathcal{V}}\sum_{u \in \mathcal{V}}\sum_{s\in\mathcal{V}}\!\!\#(\!v,\!s)\!\Big(\!\!\left[u\!=\!s\right]\!-\! \big \langle \!\mathbf{A}[u],\mathbf{B}[v]\big \rangle \!\Big)^2  \!\!\!\!\!+\! \mathcal{R}(\mathbf{A}, \mathbf{B})\nonumber\\
	=&\argmin_{\mathbf{A},\mathbf{B}}\!\!\sum_{v \in \mathcal{V}}\sum_{s\in\mathcal{V}}\!\!\#(\!v,\!s)\!\!\sum_{u \in \mathcal{V}}\!\Big(\!\!\left[u\!=\!s\right]\!-\! \big \langle \!\mathbf{A}[u],\mathbf{B}[v]\big \rangle \!\Big)^2  \!\!\!\!\!+\! \mathcal{R}(\mathbf{A}, \mathbf{B})\nonumber\\
	=&\argmin_{\mathbf{A},\mathbf{B}}\!\!\sum_{(v,s) \in \mathcal{V}^2}\!\!\#(\!v,\!s)\!\!\sum_{u \in \mathcal{V}}\!\Big(\!\!\left[u\!=\!s\right]\!-\! \big \langle \!\mathbf{A}[u],\mathbf{B}[v]\big \rangle \!\Big)^2  \!\!\!\!\!+\! \mathcal{R}(\mathbf{A}, \mathbf{B})\nonumber
	\end{align}
	\vspace{-0.4cm}
	\begin{align}
	=&\argmin_{\mathbf{A},\mathbf{B}}\!\!\sum_{\textbf{w}\in\mathcal{W}}\sum_{l=1}^{L}\sum_{\substack{|j|\leq\gamma\\j\not=0}}^{\gamma}\sum_{u \in \mathcal{V}}\!\!\Big(\left[\{w_{l+j} = u\right] \!-\! \big \langle \mathbf{A}[u],\mathbf{B}[w_{l}]\big \rangle \Big)^2 \!\!+\nonumber \\&\hspace{6.25cm}\mathcal{R}(\mathbf{A}, \mathbf{B}), \label{eq:objective_without_kernel}
\end{align}

\noindent where each $\textbf{w}\in\mathcal{V}^{L}$ indicates a random walk of length $L$ in the collection $\mathcal{W}$, $\#(v,s)$ denotes the number of appearances of the center and context pairs $(v,s)$ in the collection $\mathcal{W}$, $\left[\cdot\right]$ is the the Iverson bracket, and $\mathcal{R}(\mathbf{A}, \mathbf{B})$ is the regularization term. Note that, in the equation above, the last line follows from the fact that $\#(v,s)$ is equal to zero for the pairs which are not center-context, and the term $\sum_{\textbf{w}\in\mathcal{W}}\sum_{l=1}^{L}\sum_{|j|\leq\gamma}^{\gamma}$ traverses all center-context pairs in the collection. Matrix $\mathbf{A}$ in Eq. \eqref{eq:objective_without_kernel} indicates the embedding vectors of nodes when they are considered as \textit{centers}; those will be the embeddings that are used in the experimental evaluation. The choice of matrix $\mathbf{M}$ and the reformulation of the  objective function, offers a computational advantage during the optimization step. Moreover, such formulation also allows us to further explore \textit{kernelized} version in order to exploit possible non-linearity of the model.

\section{Kernel-based Representation Learning}\label{sec:kernel}

Most matrix factorization techniques that aim
to find latent low-dimensional representations (e.g., \cite{netmf, netsmf-www2019, hope}), adopt the \textit{Singular Value Decomposition} (SVD) provides the best approximation of the objective function stated in Eq. \eqref{eq:main_obj_func},  as long as the weight matrix is uniform \cite{svd}. Nevertheless, in our case the weight matrix is not uniform, therefore  we need the exact realization of the target matrix in order to perform SVD. To overcome this limitation, we leverage kernel functions to learn node representations via matrix factorization.

Let $(\mathsf{X}, d_X)$ be a metric space and $\mathbb{H}$ be a Hilbert space of real-valued functions defined on $\mathsf{X}$. A Hilbert space is called  \textit{reproducing kernel Hilbert space (RKHS)} if the point evaluation map over $\mathbb{H}$ is a continuous linear functional. Furthermore, a \textit{feature map} is defined as a function $\Phi:\mathsf{X}\rightarrow\mathbb{H}$ from the input space $\mathsf{X}$ into \textit{feature space} $\mathbb{H}$. Every feature map defines a \textit{kernel} $\mathsf{K}:\mathsf{X}\times\mathsf{X} \rightarrow \mathbb{R}$ as follows:
\begin{align*}
    \mathsf{K}(\mathbf{x},\mathbf{y}) := \langle \Phi(\mathbf{x}), \Phi(\mathbf{y}) \rangle && \forall (\mathbf{x},\mathbf{y}) \in \mathsf{X}^2.
\end{align*}

\noindent It can be seen that $\mathsf{K}(\cdot,\cdot)$ is symmetric and positive definite due to the properties of an inner product space. 

A function $g:\mathsf{X}\rightarrow\mathbb{R}$ is \textit{induced by} $\mathsf{K}$, if there exists $h\in\mathbb{H}$ such that $g=\langle h,\Phi(\cdot) \rangle$ for a feature vector $\Phi$ of kernel $\mathsf{K}$. Note that, this is independent of the definition of the feature map $\Phi$ and space $\mathbb{H}$ \cite{kernel_steinwart}. Let $\mathcal{I}_{\kappa} := \{g:\mathsf{X}\rightarrow\mathbb{R} \ | \ \exists h\in\mathbb{H} \ \text{s.t.} \ g=\langle h,\Phi(\cdot)\rangle\}$ be the set of induced functions by kernel $\mathsf{K}$. Then, a continuous kernel $\mathsf{K}$ on a compact metric space $(\mathsf{X},d_{\mathsf{X}})$ is \textit{universal}, if the set $\mathcal{I}_{\kappa}$ is dense in the space of all continuous real-valued functions $\mathcal{C}(\mathsf{X})$. In other words, for any function $f \in \mathcal{C}(\mathsf{X})$ and $\epsilon > 0$, there exists $g_h \in \mathcal{I}_{\kappa}$ satisfying
\begin{align*}
    \left\Vert f-g_{h} \right\Vert_{\infty} \leq \epsilon,
\end{align*}
where $g_h$ is defined as $\langle h, \Phi(\cdot) \rangle$ for some $h \in \mathbb{H}$.

\par In this paper, we consider universal kernels, since we can always find $h \in \mathbb{H}$ satisfying $| \langle h, \phi(x_i) \rangle - \alpha_{i} | \leq \epsilon$ for given $\{x_1,\ldots, x_N\}\subset\mathsf{X}$, $\{\alpha_1,\ldots,\alpha_N\}\subset \mathbb{R}$ and $\epsilon > 0$ by Proposition \ref{proposition}. If we choose $\alpha_i$'s as the entries of a row of our target matrix $\mathbf{M}$, then the elements $h$ and $\phi(x_i)$ indicate the corresponding row vectors of $\mathbf{A}$ and $\mathbf{B}$, respectively. Then, we can obtain a decomposition of the target matrix by repeating the process for each row. However, the element $h$ might not always be in the range of feature map $\Phi$; in this case, we will approximate the correct values of $\mathbf{M}$.

\begin{proposition}[Universal kernels \cite{kernel_steinwart}]\label{proposition}
    Let $(\mathsf{X}, d)$ be a compact metric space and $\mathsf{K}(\cdot, \cdot)$ be a universal kernel on $\mathsf{X}$. Then, for all compact and mutually disjoint subsets $\mathsf{S}_1,\ldots,\mathsf{S}_n \subset \mathsf{X}$, all $\alpha_1,\ldots,\alpha_n$ $\in \mathbb{R}$, and all $\epsilon > 0$,  there exists a function $g$ induced by $\mathsf{K}$ with $\norm{g}_{\infty} \leq \max_i|\alpha_i| + \epsilon$ such that
    \begin{align*}
        \norm{g_{|\mathsf{S}} - \sum_{i=1}^{n}\alpha_i\mathbb{1}_{\mathsf{S}_i}}_{\infty} \leq \epsilon, 
    \end{align*}
    where $\mathsf{S} := \bigcup_{i=1}^n\mathsf{S}_i$ and $g_{|\mathsf{S}}$ is the restriction of $g$ to $\mathsf{S}$.
    \label{prp:1}
\end{proposition}

Universal kernels also provide a guarantee for the injectivity of the feature maps, as shown in Lemma \ref{lemma:injectivity}; therefore, we can always find $y\in \mathsf{X}$, such that $\Phi(y)=h$ if $h \in \Phi(\mathsf{X})$. Otherwise, we can learn an approximate pre-image solution by using a gradient descent technique \cite{preimage_problem}.

\begin{lemma}[\cite{kernel_steinwart}]\label{lemma:injectivity}
Every feature map of a universal kernel is injective.
\end{lemma}
\begin{proof}
Let $\Phi(\cdot):\mathsf{X}\rightarrow \mathbb{H}$ be a feature vector of the kernel $\mathsf{K}(\cdot, \cdot)$. Assume that $\Phi$ is not injective, so we can find a pair of distinct elements $x, y\in \mathsf{X}$ such that $\Phi(x)=\Phi(y)$ and $x \not= y$. By Proposition \ref{proposition}, for any given $\epsilon > 0$, there exists a function $g=\langle h, \Phi(\cdot)  \rangle$ induced by $\mathsf{K}$ for some $h\in\mathbb{H}$, which satisfies
\begin{align*}
    \big\Vert  g_{|S} - ( \mathbb{1}_{S_1} - \mathbb{1}_{S_2} ) \big\Vert_{\infty} \leq \epsilon,
\end{align*}
where $\mathsf{S} := \mathsf{S}_1 \cup  \mathsf{S}_2 $ for  the compact sets $\mathsf{S}_1 = \{x\}$ and $\mathsf{S}_2 = \{y\}$. Then, we have that $|  \langle \Phi(x), h \rangle -  \langle \Phi(y), h \rangle ) | \geq 2-2\epsilon$. In other words, we obtain $\Phi(x) \not= \Phi(y)$, which contradicts our initial assumption.
\end{proof}

\begin{figure*}[t]
\centering
\includegraphics[width=0.94\textwidth]{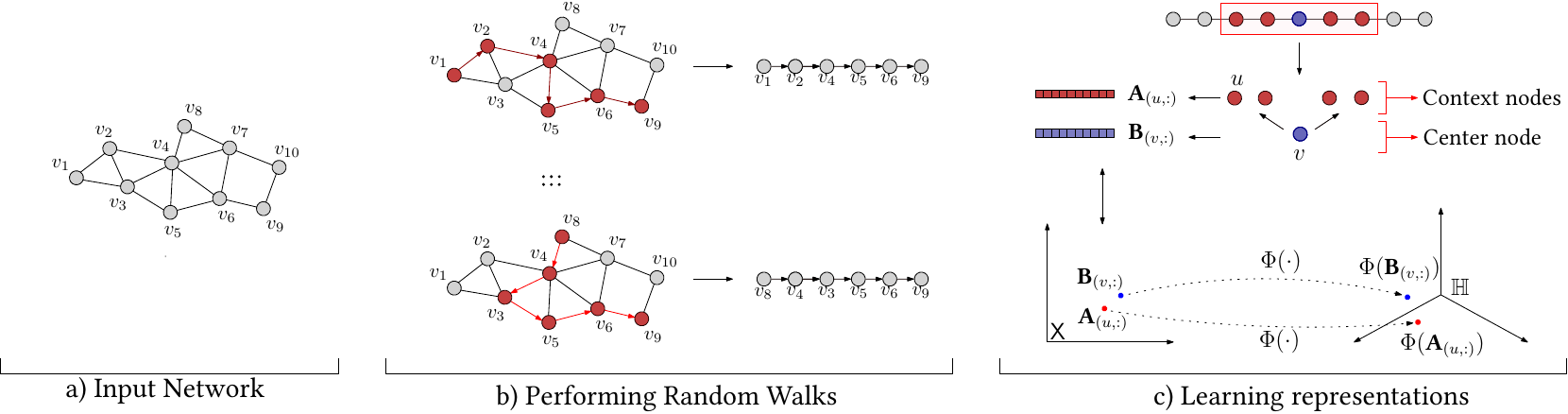}
\caption{Schematic representation of the \textsc{KernelNE} model. Node sequences are firstly generated by following a random walk strategy. By using the co-occurrences of node pairs within a certain window size, node representations are learned by optimizing their maps in the feature space.} \label{fig:main-figure}
\end{figure*}

\subsection{Single Kernel Node Representation Learning}\label{subsec:single}

Following the kernel formulation described above, we can now perform matrix factorization in the feature space by leveraging kernel functions. In particular, we can move the inner product from the input space $\textsf{X}$ to the feature space $\mathbb{H}$, by reformulating Eq. \eqref{eq:objective_without_kernel} as follows:

\begin{align}
	&\!\!\!\!\!\!\!\!\!\argmin_{\mathbf{A},\mathbf{B}}\!\!\!\sum_{\textbf{w}\in\mathcal{W}}\sum_{l=1}^{L}\sum_{\substack{|j|\leq\gamma\\j\not=0}}\sum_{u \in \mathcal{V}}\!\Big(\! \!\!\left[w_{l+j} \!\!=\! u\right]  \!\!-\! \big \langle \Phi(\mathbf{A}_{(u,:)}),\Phi(\mathbf{B}_{(w_{l},:)})\big \rangle\! \Big)^2\nonumber\\
	&\hspace{6.0cm}+ \mathcal{R}(\mathbf{A}, \mathbf{B}) \nonumber\\
	=&\argmin_{\mathbf{A},\mathbf{B}}\!\!\!\sum_{\textbf{w}\in\mathcal{W}}\sum_{l=1}^{L}\sum_{\substack{|j|\leq\gamma\\j\not=0}}\sum_{u \in \mathcal{V}}\!\Big(\!\!\left[w_{l+j} \!\!=\! u\right]  \!-\! \mathsf{K}\big(\mathbf{A}_{(u,:)},\mathbf{B}_{(w_{l},:)}\big)\! \Big)^2\nonumber\\ 
	& \hspace{6.0cm} +\mathcal{R}(\mathbf{A}, \mathbf{B}). \label{eq:main_obj_kernel}
\end{align}

\noindent In this way, we obtain a kernelized matrix factorization model for node embeddings based on random walks. For the numerical evaluation of our method, we use the following universal kernels \cite{universal_kernels, kernel_steinwart}:

\begin{align*}
    \mathsf{K}_{G}(\mathbf{x},\mathbf{y}) & = \exp\left(\frac{-\left\Vert \mathbf{x}-\mathbf{y}\right\Vert^2}{\sigma^2}\right) && \sigma \in \mathbb{R}\\
    \mathsf{K}_{S}(\mathbf{x},\mathbf{y}) & = \frac{1}{\left( 1 + \left\Vert \mathbf{x}-\mathbf{y} \right\Vert^2\right)^{\sigma}} && \sigma \in \mathbb{R}_+,
\end{align*}

\noindent where $\mathsf{K}_{G}$ and $\mathsf{K}_{S}$ correspond to the \textit{Gaussian} and \textit{Schoenberg} kernels respectively. We will refer to the proposed kernel-based node embeddings methodology as \textsc{KernelNE} (the two different kernels will be denoted by \textsc{Gauss} and \textsc{Sch}). A schematic representation of the basic components of the proposed model is given in Fig. \ref{fig:main-figure}.

\subsubsection{Model Optimization} 

The estimation problem for both parameters $\mathbf{A}$ and $\mathbf{B}$ is, unfortunately, non-convex. Nevertheless,  when we consider each parameter separately by fixing the other one, it turns into a convex problem. By taking advantage of this property, we employ \textit{Stochastic Gradient Descent} (SGD) \cite{sgd} in the optimization step of each embedding matrix. Note that, for each context node $w_{l+j}$ in Eq. \eqref{eq:main_obj_kernel}, we have to compute the gradient for each node $u\in \mathcal{V}$, which is computationally intractable. However, Eq. \eqref{eq:main_obj_kernel} can be divided into two parts with respect to the values of $\left[w_{l+j} = u\right]\!\in\!\{0,1\}$, as follows:

\begin{align*}
    &\sum_{u \in \mathcal{V}}\Big(\left[w_{l+j} \!\!=\! u\right] - \mathsf{K}\big (\mathbf{A}_{(u,:)},\mathbf{B}_{(w_{l},:)}\big) \Big)^2 \\
    =&\Big(\!1 \!-\! \mathsf{K}\big(\mathbf{A}_{(u^+,:)},\mathbf{B}_{(w_{l},:)}\big) \!\!\Big)^2 \!\!+\!\!\!\!\!\!\!\!\! \sum_{u^{-}\in\mathcal{V} \backslash \{w_{l+j}\} }\!\!\!\!\!\!\!\!\!\big( \mathsf{K}\big(\mathbf{A}_{(u^-,:)},\mathbf{B}_{(w_{l},:)}\big) \big)^2\\
    \approx&\Big(\! 1 \!\!-\!\! \mathsf{K}\big(\mathbf{A}_{(u^+,:)},\mathbf{B}_{(w_{l},:)}\!\big)\!\!\Big)^2 \!\!\!\!\!+\! \large|\mathcal{V}\backslash \{w_{l+j}\}\large|\!\!\!\!\!\underset{u^{-}\sim \ p^- }{\!\!\!\!\!\mathbb{E}}\!\!\!\!\!\!\!\!\!\big[ \mathsf{K}(\mathbf{A}_{(u^-,:)},\mathbf{B}_{(w_{l},:)}) \big]^2\\
    =&\underbrace{\mystrut{2.0ex}\big( 1 - \mathsf{K}\big(\mathbf{A}_{(u^+,:)},\mathbf{B}_{(w_{l},:)}\big)^2}_{\text{positive sample}} \!\!+ \underbrace{k\!\!\!\!\underset{u^{-}\sim p^- }{\mathbb{E}}\big[ \mathsf{K}(\mathbf{A}_{(u^-,:)},\mathbf{B}_{
    (w_{l},:)}) \big]^2}_{\text{negative sample}},
\end{align*}

\noindent where $k:=|\mathcal{V}|-1$ and $u^+ := w_{l+j}$. To this end, we apply \textit{negative sampling} \cite{word2vec} which is a variant of \textit{noise-contrastive estimation} \cite{nce}, proposed as an alternative to solve the computational problem of hierarchical softmax. For each context node $u^+ \in \mathcal{C}_{\textbf{w}}(w_{l})$, we sample $k$ negative instances $u^-$ from the noise distribution $p^-$. Then, we can rewrite the  objective function Eq. \eqref{eq:main_obj_kernel} in the following way:
\begin{align}
\mathcal{F}_{S} :=& \argmin_{\mathbf{A},\mathbf{B}}\!\sum_{\boldsymbol{w}\in\mathcal{W}}\sum_{l=1}^L\sum_{\substack{|j|\leq\gamma\\j\not=0}}\!\bigg( \big(1 \!\!- \mathsf{K}(\mathbf{A}_{(w_{l+j},:)},\mathbf{B}_{(w_l,:)})\big)^2 \nonumber\\
&\hspace{0.6cm}+\sum_{\substack{r=1\\u_r^{-} \sim p^-}}^k\!\!\mathsf{K}\big(\mathbf{A}_{(u_r^-,:)},\!\mathbf{B}_{(w_l,:)} \big)^2 \bigg) + \mathcal{R}(\mathbf{A},\mathbf{B}).\label{eq:kernelne}
\end{align}

\noindent Equation \eqref{eq:kernelne} corresponds to the objective function of the proposed \textsc{KernelNE} model. In the following subsection, we will study how this model could be further extended to leverage multiple kernels.

\subsection{Multiple Kernel Node Representation Learning}
Selecting a proper kernel function $\mathsf{K}(\cdot,\cdot)$ and the corresponding parameters (e.g., the bandwidth of a Gaussian kernel) is a critical task during the learning phase.  Nevertheless, choosing a single kernel function might impose potential bias, causing limitations on the performance of the model. Having the ability to properly utilize multiple kernels could increase the expressiveness of the model, capturing different notions of similarity among embeddings \cite{mkl-jmlr11}. Besides, learning how to combine such kernels, might further improve the performance of the underlying model. In particular, given a set of base kernels $\{\mathsf{K}_i\}_{i=1}^{K}$,  we aim to find an optimal way to combine them, as follows:

\begin{equation*}
    \mathsf{K}^{\mathbf{c}} (\mathbf{x}, \mathbf{y}) = f_\mathbf{c} \big(\{\mathsf{K}_i(\mathbf{x}, \mathbf{y})\}_{i=1}^{K}  | \mathbf{c} \big),
\end{equation*}

\noindent where the combination function $f_{\mathbf{c}}$ is parameterized on $\mathbf{c} \in \mathbb{R}^K$ that indicates kernel weights. Due to the generality of the multiple kernel learning framework, $f_{\mathbf{c}}$ can be either a linear or nonlinear function. 

\par In this paragraph, we examine how to further strengthen the proposed kernelized weighted matrix factorization, by \textit{linearly} combining multiple kernels. Let $\mathsf{K}_1, \ldots, \mathsf{K}_K$ be a set of kernel functions satisfying the properties presented in the previous paragraph. Then, we can restate the objective function as follows:
\begin{align}
&\mathcal{F}_{M}\!:=\!\argmin_{\mathbf{A},\mathbf{B},\textbf{c}}\!\!\!\sum_{\boldsymbol{w}\in\mathcal{W}}\sum_{l=1}^L\sum_{\substack{|j|\leq \gamma\\j\not=0}}\!\Big( \!1 \!\!- \!\!\sum_{i=1}^{K}c_i\mathsf{K}_i(\mathbf{A}_{(w_{l+j},:)},\mathbf{B}_{(w_l,:)})\big)^2 \nonumber\\ 
&\hspace{3.0cm}+\!\!\!\!\sum_{\substack{r=1\\u_r^{-} \sim p^-}}^k\!\!\!\!\!\big(\sum_{i=1}^{K}\!c_i\mathsf{K}_i(\!\mathbf{A}_{(u_r^-\!\!,:)},\!\mathbf{B}_{(w_l,:)})\big)^2 \Big)\nonumber\\
&\hspace{3.0cm}+ \frac{\lambda}{2}\left( \Vert \mathbf{A} \Vert_F^2 \!+  \Vert \mathbf{B} \Vert_F^2\!\right) + \frac{\beta}{2}\Vert \textbf{c} \Vert_2^2,
\label{eq:mkernelne}
\end{align}
\noindent where $\textbf{c}:=[c_1,\ldots,c_K]^{\top}\in \mathbb{R}^K$. Here, we introduce an additional parameter $c_i$ representing the contribution of the corresponding kernel $\textsf{K}_i$. $\beta>0$ is a trade-off parameter, and similarly, the coefficients $c_1,\ldots,c_K$ are optimized by fixing the remaining model parameters $\mathbf{A}$ and $\mathbf{B}$. Equation \eqref{eq:mkernelne} corresponds to the objective function of the proposed multiple kernel learning model \textsc{MKernelNE}. Unlike the common usage of multiple kernels methods \cite{mkl-jmlr11}, here we do not constrain the coefficients to be non-negative. We interpret each entry of the target matrix as a linear combination of inner products of different kernels' feature maps. As discussed in the previous sections, our main intuition relies on obtaining more expressive embeddings by projecting the factorization step into a higher dimensional space.

\par Algorithm \ref{alg:multikernel} summarizes the pseudocode of the proposed approach. For a given collection of random walks $\mathcal{W}$, we first determine the center-context node pairs $(w_{l}, w_{l+j})$ in the node sequences. Recall that, for each \textit{center} node in a walk, its surrounding nodes within a certain distance $\gamma$,  define its \textit{context}. Furthermore, the corresponding embedding vectors $\mathbf{B}_{(w_{l},:)}$ of center node $w_{l}$ and $\mathbf{A}_{(w_{l+j},:)}$ of context $w_{l+j}$ are updated by following the rules which we describe in detail below. Note that, we obtain two different representations, $\mathbf{A}_{(v,:)}$ and $\mathbf{B}_{(v,:)}$, for each node $v\in\mathcal{V}$ since the node can have either a center or context role in the walks.

\par The gradients in Alg. \ref{alg:updateemb} are given below. For notation simplicity, we denote each $\mathsf{K}_i(\mathbf{A}_{(u,:)},\mathbf{B}_{(v,:)})$ by $\mathsf{K}_i(u,v)$.
\begin{align*}
&\nabla_{\mathbf{A}_{(x,:)}}\mathcal{F}_M \!\!:=\!
\left[x \!=\! u\right]\!\big(\!\!-\!\!2\!\sum_{i=1}^{K}\!c_i\nabla_{\mathbf{A}_{(x,:)}}\!\!\mathsf{K}_i(x,v)\!\big)\!\big(\!1 \!\!-\!\! \sum_{j=1}^{K}\!c_i\mathsf{K}_i(x,v)\!\big)\\
& \hspace{0.6cm}+2\!\!\!\!\sum_{\substack{r=1\\u_r^{-} \sim p^-}}^k\!\!\!\!\left[x \!\!=\! u_r^-\right]\!\!\big(\sum_{i=1}^{K}c_i\underset{\mathbf{A}_{(u_r^-,:)}}{\nabla}\!\!\!\!\!\mathsf{K}_i(u_r^-,v)\mathsf{K}_i(u_r^-,v)\big) \!+ \!\lambda\mathbf{A}_{(x,:)}\\
&\nabla_{\mathbf{B}_{(v,:)}}\mathcal{F}_M \!\!:=\!-2\big(\!\sum_{i=1}^{K}\!c_i\nabla_{\mathbf{B}_{(v,:)}}\mathsf{K}_i(u,v)\!\big)\!\big(\!1 \!\!-\!\! \sum_{i=1}^{K}\!c_i\mathsf{K}_i(u,v)\!\big) \\
& \hspace{1.6cm}+2\!\!\!\!\!\sum_{\substack{r=1\\u_r^{-} \sim p^-}}^k\!\!\!\!\!\big(\sum_{i=1}^{K}c_i\nabla_{\mathbf{B}_{(v,:)}}\mathsf{K}_i(u_r^-,v)\mathsf{K}_i(u_r^-,v)\big) \!+\! \lambda\mathbf{B}_{(v,:)}\\
&\nabla_{c_{t}}\mathcal{F}_M \!\!=\!-2 \ \big(\!1 \!\!-\!\! \sum_{i=1}^{K}\!c_i\mathsf{K}_i(u,v)\!\big)\mathsf{K}_t(u,v) \\
& \hspace{2.6cm}+\!2\!\!\!\!\sum_{\substack{r=1\\u_r^{-} \sim p^-}}^k\!\!\!\!\big(\sum_{i=1}^{K}\!c_i\mathsf{K}_i(u_r^-, v)\big)\mathsf{K}_t(u_r^-,v) + \beta c_{t}
\end{align*}

\subsection{Complexity Analysis}
For the generation of walks, the biased random walk strategy proposed in \textsc{Node2Vec} is used, which  can be performed in $\mathcal{O}(|\mathcal{W}|\cdot L)$ steps \cite{node2vec} for the pre-computed transition probabilities, where $L$ indicates the walk length and $\mathcal{W}$ denotes the set of walks. For the algorithm's learning procedure, we can carry out Line $5$ of Algorithm \ref{alg:multikernel} at most $2\gamma\cdot |\mathcal{W}|\cdot L$ times for each center-context pair, where $\gamma$ represents the window size. The dominant operation in Algorithm \ref{alg:updateemb} is the multiplication operation of update rules in Lines $6$, $7$, and $9$; the running time can be bounded by $\mathcal{O}(k \cdot K \cdot d)$, where $k$ is the number of negative samples generated per center-context pair, $K$ is the number of kernels, and $d$ is the representation size. To sum up, the overall running time of the proposed approach can be bounded by $\mathcal{O}(\gamma \cdot |\mathcal{W}|\cdot L \cdot K \cdot d \cdot k)$ steps.

\begin{algorithm}[t]
\begin{algorithmic}[1]
\REQUIRE Graph $G=(\mathcal{V},\mathcal{E})$\\ \hspace{0.45cm} Representation size $d$\\
\hspace{0.45cm} Set of walks $\mathcal{W}$\\
\hspace{0.45cm} Window size $\gamma$\\
\hspace{0.45cm} Kernel function $\mathsf{K}$ \\
\hspace{0.45cm} Kernel parameter(s) $\sigma$\\
\ENSURE Embedding matrix $\mathbf{A}$
\STATE  Initialize matrices $\mathbf{A}$,
$\mathbf{B}\in\mathbb{R}^{n\times d}$ \\
\textcolor{gray}{\texttt{\textbf{/* Extract center-context node pairs */}}}
\FORALL{$\textbf{w}=(w_1,\ldots,w_L) \in \mathcal{W}$}
\FOR{$l \gets 1$ \TO $L$}
\FOR{$j\not=0 \gets$ $-\gamma$ \TO $\gamma$}
\STATE $\mathbf{A}, \mathbf{B}, \mathbf{c} \leftarrow \textsc{UpdateEmb}(\mathbf{A}, \mathbf{B},\textbf{c}, w_l, w_{l+j}, \mathsf{K}, \sigma)$
\ENDFOR
\ENDFOR
\ENDFOR
\end{algorithmic}
\caption{\textsc{MKernelNE}}\label{alg:multikernel}
\end{algorithm}
\begin{algorithm}[t]
\begin{algorithmic}[1]
\REQUIRE Graph $G=(\mathcal{V},\mathcal{E})$\\ 
\hspace{0.45cm} Embedding matrices $\mathbf{A}$ and $\mathbf{B}$\\
\hspace{0.45cm} Kernel coefficients $\mathbf{c}=(c_1,\ldots,c_K)$\\
\hspace{0.45cm} Kernel function(s) $\mathsf{K}$ \\
\hspace{0.45cm} Center and context nodes $v$ and $u$\\
\hspace{0.45cm} Learning rate $\eta$\\
\hspace{0.45cm} Distribution for generating negative samples $p^-$\\
\ENSURE Embedding matrix $\mathbf{A}$
\STATE node\_list $\gets$ $[u]$ \\
\textcolor{gray}{\texttt{\textbf{/* Extract negative samples */}}}
\FOR{$s \gets 1$ \TO $k$}
\STATE node\_list $\gets$ \textsc{SampleNode}$(p^-)$ 
\ENDFOR \\
{\small{\textcolor{gray}{\texttt{\textbf{/* Update embedding vectors */}}}}}
\FOR{\textbf{each} $x$ in node\_list}
\STATE $\mathbf{A}_{(x,:)} \leftarrow \mathbf{A}_{(x,:)} -  \eta\nabla_{\mathbf{A}_{(x,:)}} \mathcal{F}_M$
\STATE $\mathbf{B}_{(v,:)} \leftarrow \mathbf{B}_{(v,:)} -  \eta\nabla_{\mathbf{B}_{(v,:)}} \mathcal{F}_M$ \\
{\small{\textcolor{gray}{\texttt{\textbf{/* Update individual kernel weights */}}}}}
\IF{number of kernels $>$ 1}
\STATE  $\mathbf{c} \leftarrow \mathbf{c} -  \eta\nabla_{c} \mathcal{F}_M$
\ENDIF
\ENDFOR
\end{algorithmic}
\caption{\textsc{UpdateEmb}}\label{alg:updateemb}
\end{algorithm}
\section{Experiments}
This section presents the experimental set-up details, the datasets, and the baseline methods used in the evaluation. The performance of the proposed single and multiple kernel models is examined for  node classification and link prediction tasks on various real-world datasets. The experiments have been performed on a computer with 16Gb RAM.

\subsection{Baseline Methods}
We consider nine baseline models to compare the performance of our approach. (i) \textsc{DeepWalk} \cite{deepwalk} performs uniform random walks to generate the contexts of nodes; then, the \textsc{Skip-Gram} model \cite{word2vec} is used to learn node embeddings. (ii) \textsc{Node2Vec} \cite{node2vec} combines \textsc{Skip-Gram} with biased random walks, using two extra parameters that control the walk in order to simulate a \textit{BFS} or \textit{DFS} exploration. In the experiments, we set those parameters to $1.0$. In our approach we sample context nodes using this biased random walk strategy. (iii) \textsc{LINE} \cite{line} learns nodes embeddings relying on first- and second-order proximity information of nodes. (iv) \textsc{HOPE} \cite{hope} is a matrix factorization approach aiming at capturing similarity patterns based on a higher-order node similarity measure. In the experiments, we consider the \textit{Katz} index, since it demonstrates the best performance among other proximity indices. (v) \textsc{NetMF} \cite{netmf} targets to factorize the matrix approximated by pointwise mutual information of center and context pairs. The experiments have been conducted for large window sizes ($\gamma=10$) due to its good performance. (vi) \textsc{VERSE} \cite{verse-www18} learns the embedding vectors by optimizing similarities among nodes. As suggested by the authors, we set $\alpha=0.85$ for the value of the hyper-parameter in the experiments. (vii) \textsc{ProNE} \cite{prone} learns the representations by relying on an efficient sparse matrix factorization and the extracted embeddings are improved with spectral propagation operations. (viii) \textsc{GEMSEC} \cite{gemsec_asonam_bedenek} leverages the community structure of real-world graphs, learning node embeddings and the cluster assignments simultaneously. We have used the best performing number of cluster value from the set $\{5,10, 15, 25, 50, 75, 100\}$. (ix) Lastly, \textsc{M-NMF} \cite{mnmf_aaai17_wang} extracts node embeddings under a  modularity-based community detection framework based on non-negative matrix factorization. We have observed that the algorithm poses good performance by setting its parameters $\alpha=0.1$ and $\beta=5$. We performed parameter tuning for the number of communities using values from the following set: $\{5, 15, 20, 25, 50, 75, 100\}$.

Those baseline methods are compared against instances of \textsc{KernelNE} and \textsc{MKernelNE} using different kernel functions (\textsc{Gauss} and \textsc{Sch}).

\begin{table}
\caption{Statistics of networks $|\mathcal{V}|$: number of nodes, $|\mathcal{E}|$: number of edges, $|\mathcal{K}|$: number of labels and $|\mathcal{C}c|$: number of connected components.}
\label{tab:networks}
\resizebox{\columnwidth}{!}{%
\begin{tabular}{rccccc}
\toprule
\multicolumn{1}{l}{} & \textbf{$|\mathcal{V}|$} & \textbf{$|\mathcal{E}|$} & \textbf{$|\mathcal{K}|$} & \textbf{$|\mathcal{C}c|$} & \textbf{Avg. Degree}  \\\midrule
\textsl{CiteSeer} & 3,312 & 4,660 & 6 & 438 & 2.81  \\
\textsl{Cora} & 2,708 & 5,278 & 7 & 78 & 3.90   \\
\textsl{DBLP} & 27,199 & 66,832 & 4 & 2,115 & 4.91  \\
\textsl{PPI} & 3,890 & 38,739 & 50 & 35 & 19.92 \\\midrule
\textsl{AstroPh} & 17,903 & 19,7031 & - & 1 & 22.01  \\
\textsl{HepTh} & 8,638 & 24,827 & - & 1 & 5.74 \\
\textsl{Facebook} & 4,039 & 88,234 & - & 1 & 43.69 \\
\textsl{Gnutella} & 8,104 & 26,008 & - & 1 & 6.42 \\\bottomrule
\end{tabular}%
}
\end{table}

\begin{table*}[]
\caption{Node classification task for varying training sizes on \textsl{Citeseer}. For each method, the rows show the Micro-$F_1$ and Macro-$F_1$ scores, respectively.}
\label{tab:classification_citeseer}
\centering
\resizebox{\textwidth}{!}{%
\begin{tabular}{crcccccccccccccccccc}\toprule
& & \textbf{1\%} & \textbf{2\%} & \textbf{3\%} & \textbf{4\%} & \textbf{5\%} & \textbf{6\%} & \textbf{7\%} & \textbf{8\%} & \textbf{9\%} & \textbf{10\%} & \textbf{20\%} & \textbf{30\%} & \textbf{40\%} & \textbf{50\%} & \textbf{60\%} & \textbf{70\%} & \textbf{80\%} & \textbf{90\%} \\\midrule
& & 0.367 & 0.421 & 0.445 & 0.462 & 0.475 & 0.488 & 0.496 & 0.502 & 0.508 & 0.517 & 0.550 & 0.569 & 0.580 & 0.587 & 0.593 & 0.596 & 0.597 & 0.597 \\
& \multirow{-2}{*}{$\textsc{DeepWalk}$} & \cellcolor[HTML]{EFEFEF}0.313 & \cellcolor[HTML]{EFEFEF}0.374 & \cellcolor[HTML]{EFEFEF}0.402 & \cellcolor[HTML]{EFEFEF}0.419 & \cellcolor[HTML]{EFEFEF}0.434 & \cellcolor[HTML]{EFEFEF}0.446 & \cellcolor[HTML]{EFEFEF}0.455 & \cellcolor[HTML]{EFEFEF}0.461 & \cellcolor[HTML]{EFEFEF}0.467 & \cellcolor[HTML]{EFEFEF}0.476 & \cellcolor[HTML]{EFEFEF}0.506 & \cellcolor[HTML]{EFEFEF}0.524 & \cellcolor[HTML]{EFEFEF}0.534 & \cellcolor[HTML]{EFEFEF}0.540 & \cellcolor[HTML]{EFEFEF}0.544 & \cellcolor[HTML]{EFEFEF}0.547 & \cellcolor[HTML]{EFEFEF}0.547 & \cellcolor[HTML]{EFEFEF}0.546 \\
& & 0.404 & 0.451 & 0.475 & 0.491 & 0.504 & 0.514 & 0.523 & 0.529 & 0.537 & 0.543 & 0.571 & 0.583 & 0.591 & 0.595 & 0.600 & 0.600 & 0.601 & 0.603 \\
& \multirow{-2}{*}{$\textsc{Node2Vec}$} & \cellcolor[HTML]{EFEFEF}0.342 & \cellcolor[HTML]{EFEFEF}0.397 & \cellcolor[HTML]{EFEFEF}0.424 & \cellcolor[HTML]{EFEFEF}0.443 & \cellcolor[HTML]{EFEFEF}0.456 & \cellcolor[HTML]{EFEFEF}0.466 & \cellcolor[HTML]{EFEFEF}0.476 & \cellcolor[HTML]{EFEFEF}0.483 & \cellcolor[HTML]{EFEFEF}0.489 & \cellcolor[HTML]{EFEFEF}0.497 & \cellcolor[HTML]{EFEFEF}0.525 & \cellcolor[HTML]{EFEFEF}0.535 & \cellcolor[HTML]{EFEFEF}0.542 & \cellcolor[HTML]{EFEFEF}0.546 & \cellcolor[HTML]{EFEFEF}0.549 & \cellcolor[HTML]{EFEFEF}0.549 & \cellcolor[HTML]{EFEFEF}0.549 & \cellcolor[HTML]{EFEFEF}0.550 \\
& & 0.253 & 0.300 & 0.332 & 0.353 & 0.369 & 0.384 & 0.395 & 0.407 & 0.412 & 0.418 & 0.459 & 0.476 & 0.487 & 0.494 & 0.500 & 0.505 & 0.507 & 0.512 \\
& \multirow{-2}{*}{$\textsc{LINE}$} & \cellcolor[HTML]{EFEFEF}0.183 & \cellcolor[HTML]{EFEFEF}0.243 & \cellcolor[HTML]{EFEFEF}0.280 & \cellcolor[HTML]{EFEFEF}0.303 & \cellcolor[HTML]{EFEFEF}0.321 & \cellcolor[HTML]{EFEFEF}0.334 & \cellcolor[HTML]{EFEFEF}0.346 & \cellcolor[HTML]{EFEFEF}0.357 & \cellcolor[HTML]{EFEFEF}0.363 & \cellcolor[HTML]{EFEFEF}0.367 & \cellcolor[HTML]{EFEFEF}0.407 & \cellcolor[HTML]{EFEFEF}0.423 & \cellcolor[HTML]{EFEFEF}0.434 & \cellcolor[HTML]{EFEFEF}0.440 & \cellcolor[HTML]{EFEFEF}0.445 & \cellcolor[HTML]{EFEFEF}0.448 & \cellcolor[HTML]{EFEFEF}0.450 & \cellcolor[HTML]{EFEFEF}0.454 \\
& & 0.198 & 0.197 & 0.201 & 0.204 & 0.205 & 0.207 & 0.211 & 0.208 & 0.212 & 0.216 & 0.235 & 0.253 & 0.265 & 0.276 & 0.288 & 0.299 & 0.304 & 0.316 \\
& \multirow{-2}{*}{$\textsc{HOPE}$} & \cellcolor[HTML]{EFEFEF}0.064 & \cellcolor[HTML]{EFEFEF}0.060 & \cellcolor[HTML]{EFEFEF}0.061 & \cellcolor[HTML]{EFEFEF}0.065 & \cellcolor[HTML]{EFEFEF}0.065 & \cellcolor[HTML]{EFEFEF}0.066 & \cellcolor[HTML]{EFEFEF}0.070 & \cellcolor[HTML]{EFEFEF}0.068 & \cellcolor[HTML]{EFEFEF}0.072 & \cellcolor[HTML]{EFEFEF}0.075 & \cellcolor[HTML]{EFEFEF}0.099 & \cellcolor[HTML]{EFEFEF}0.121 & \cellcolor[HTML]{EFEFEF}0.136 & \cellcolor[HTML]{EFEFEF}0.150 & \cellcolor[HTML]{EFEFEF}0.164 & \cellcolor[HTML]{EFEFEF}0.178 & \cellcolor[HTML]{EFEFEF}0.186 & \cellcolor[HTML]{EFEFEF}0.202 \\
&  & {\ul 0.434} & {0.476} & {0.489} & {0.512} & {0.517} & {0.532} & {0.538} & {0.536} & {0.551} & {0.554} & {0.576} & {0.590} & {0.594} & {0.607} & {0.614} & {0.615} & {0.616} & {\textbf{0.620}} \\
 & \multirow{-2}{*}{$\textsc{VERSE}$} & \cellcolor[HTML]{EFEFEF}{\textbf{0.377}} & \cellcolor[HTML]{EFEFEF}{0.423} & \cellcolor[HTML]{EFEFEF}{0.446} & \cellcolor[HTML]{EFEFEF}{0.462} & \cellcolor[HTML]{EFEFEF}{0.466} & \cellcolor[HTML]{EFEFEF}{0.484} & \cellcolor[HTML]{EFEFEF}{0.490} & \cellcolor[HTML]{EFEFEF}{0.492} & \cellcolor[HTML]{EFEFEF}{0.505} & \cellcolor[HTML]{EFEFEF}{0.507} & \cellcolor[HTML]{EFEFEF}{0.526} & \cellcolor[HTML]{EFEFEF}{0.545} & \cellcolor[HTML]{EFEFEF}{0.545} & \cellcolor[HTML]{EFEFEF}{0.558} & \cellcolor[HTML]{EFEFEF}{0.563} & \cellcolor[HTML]{EFEFEF}{0.564} & \cellcolor[HTML]{EFEFEF}{0.565} & \cellcolor[HTML]{EFEFEF}{\textbf{0.567}} \\
 &  & {0.244} & {0.287} & {0.344} & {0.351} & {0.376} & {0.394} & {0.417} & {0.434} & {0.444} & {0.452} & {0.514} & {0.530} & {0.547} & {0.554} & {0.563} & {0.562} & {0.554} & {0.558} \\
 & \multirow{-2}{*}{$\textsc{ProNE}$} & \cellcolor[HTML]{EFEFEF}{0.172} & \cellcolor[HTML]{EFEFEF}{0.225} & \cellcolor[HTML]{EFEFEF}{0.275} & \cellcolor[HTML]{EFEFEF}{0.294} & \cellcolor[HTML]{EFEFEF}{0.313} & \cellcolor[HTML]{EFEFEF}{0.336} & \cellcolor[HTML]{EFEFEF}{0.358} & \cellcolor[HTML]{EFEFEF}{0.381} & \cellcolor[HTML]{EFEFEF}{0.383} & \cellcolor[HTML]{EFEFEF}{0.395} & \cellcolor[HTML]{EFEFEF}{0.454} & \cellcolor[HTML]{EFEFEF}{0.470} & \cellcolor[HTML]{EFEFEF}{0.488} & \cellcolor[HTML]{EFEFEF}{0.498} & \cellcolor[HTML]{EFEFEF}{0.502} & \cellcolor[HTML]{EFEFEF}{0.502} & \cellcolor[HTML]{EFEFEF}{0.497} & \cellcolor[HTML]{EFEFEF}{0.502}\\
& & 0.328 & 0.401 & 0.445 & 0.473 & 0.492 & 0.507 & 0.517 & 0.525 & 0.533 & 0.538 & 0.567 & 0.579 & 0.586 & 0.590 & 0.592 & 0.594 & 0.599 & 0.601 \\
 & \multirow{-2}{*}{$\textsc{NetMF}$} & \cellcolor[HTML]{EFEFEF}0.264 & \cellcolor[HTML]{EFEFEF}0.346 & \cellcolor[HTML]{EFEFEF}0.392 & \cellcolor[HTML]{EFEFEF}0.421 & \cellcolor[HTML]{EFEFEF}0.440 & \cellcolor[HTML]{EFEFEF}0.454 & \cellcolor[HTML]{EFEFEF}0.466 & \cellcolor[HTML]{EFEFEF}0.474 & \cellcolor[HTML]{EFEFEF}0.481 & \cellcolor[HTML]{EFEFEF}0.487 & \cellcolor[HTML]{EFEFEF}0.516 & \cellcolor[HTML]{EFEFEF}0.528 & \cellcolor[HTML]{EFEFEF}0.535 & \cellcolor[HTML]{EFEFEF}0.538 & \cellcolor[HTML]{EFEFEF}0.540 & \cellcolor[HTML]{EFEFEF}0.542 & \cellcolor[HTML]{EFEFEF}0.547 & \cellcolor[HTML]{EFEFEF}0.548\\
 &  & 0.337 & 0.384 & 0.415 & 0.439 & 0.449 & 0.459 & 0.475 & 0.479 & 0.484 & 0.491 & 0.517 & 0.532 & 0.538 & 0.545 & 0.551 & 0.552 & 0.556 & 0.558 \\
 & \multirow{-2}{*}{$\textsc{GEMSEC}$} & \cellcolor[HTML]{EFEFEF}0.288 & \cellcolor[HTML]{EFEFEF}0.339 & \cellcolor[HTML]{EFEFEF}0.376 & \cellcolor[HTML]{EFEFEF}0.400 & \cellcolor[HTML]{EFEFEF}0.412 & \cellcolor[HTML]{EFEFEF}0.422 & \cellcolor[HTML]{EFEFEF}0.435 & \cellcolor[HTML]{EFEFEF}0.439 & \cellcolor[HTML]{EFEFEF}0.445 & \cellcolor[HTML]{EFEFEF}0.451 & \cellcolor[HTML]{EFEFEF}0.476 & \cellcolor[HTML]{EFEFEF}0.488 & \cellcolor[HTML]{EFEFEF}0.494 & \cellcolor[HTML]{EFEFEF}0.497 & \cellcolor[HTML]{EFEFEF}0.501 & \cellcolor[HTML]{EFEFEF}0.499 & \cellcolor[HTML]{EFEFEF}0.502 & \cellcolor[HTML]{EFEFEF}0.501\\
&  & 0.222 & 0.264 & 0.295 & 0.317 & 0.338 & 0.340 & 0.364 & 0.370 & 0.376 & 0.382 & 0.423 & 0.439 & 0.448 & 0.449 & 0.455 & 0.456 & 0.461 & 0.457 \\
\multirow{-18}{*}{\rotatebox{90}{Baselines}} & \multirow{-2}{*}{$\textsc{M-NMF}$} & \cellcolor[HTML]{EFEFEF}0.109 & \cellcolor[HTML]{EFEFEF}0.161 & \cellcolor[HTML]{EFEFEF}0.205 & \cellcolor[HTML]{EFEFEF}0.229 & \cellcolor[HTML]{EFEFEF}0.255 & \cellcolor[HTML]{EFEFEF}0.261 & \cellcolor[HTML]{EFEFEF}0.285 & \cellcolor[HTML]{EFEFEF}0.293 & \cellcolor[HTML]{EFEFEF}0.301 & \cellcolor[HTML]{EFEFEF}0.306 & \cellcolor[HTML]{EFEFEF}0.353 & \cellcolor[HTML]{EFEFEF}0.370 & \cellcolor[HTML]{EFEFEF}0.380 & \cellcolor[HTML]{EFEFEF}0.381 & \cellcolor[HTML]{EFEFEF}0.389 & \cellcolor[HTML]{EFEFEF}0.390 & \cellcolor[HTML]{EFEFEF}0.394 & \cellcolor[HTML]{EFEFEF}0.390
\\\midrule
& & 0.422 & 0.465 & 0.490 & 0.504 & 0.518 & 0.528 & 0.535 & 0.543 & 0.551 & 0.555 & {\ul 0.588} & \textbf{0.600} & {\ul 0.607} & {\ul 0.611} & {\textbf{0.616}} & {\ul 0.616} & {\ul 0.619} & \textbf{0.620} \\
& \multirow{-2}{*}{$\textsc{Gauss}$} & \cellcolor[HTML]{EFEFEF}0.367 & \cellcolor[HTML]{EFEFEF}0.415 & \cellcolor[HTML]{EFEFEF}0.440 & \cellcolor[HTML]{EFEFEF}0.456 & \cellcolor[HTML]{EFEFEF}0.471 & \cellcolor[HTML]{EFEFEF}0.479 & \cellcolor[HTML]{EFEFEF}0.488 & \cellcolor[HTML]{EFEFEF}0.496 & \cellcolor[HTML]{EFEFEF}0.504 & \cellcolor[HTML]{EFEFEF}0.508 & \cellcolor[HTML]{EFEFEF}\textbf{0.540} & \cellcolor[HTML]{EFEFEF}\textbf{0.551} & \cellcolor[HTML]{EFEFEF}\textbf{0.558} & \cellcolor[HTML]{EFEFEF}\textbf{0.562} & \cellcolor[HTML]{EFEFEF}\textbf{0.564} & \cellcolor[HTML]{EFEFEF}\textbf{0.566} & \cellcolor[HTML]{EFEFEF}\textbf{0.567} & \cellcolor[HTML]{EFEFEF}{\ul 0.566} \\
& & 0.441 & 0.497 & 0.517 & 0.531 & 0.539 & 0.544 & 0.551 & 0.555 & 0.558 & 0.561 & 0.580 & 0.591 & 0.597 & 0.602 & 0.608 & 0.610 & 0.614 & 0.609 \\
\multirow{-4}{*}{\rotatebox{90}{\scriptsize\textsc{KernelNE}}} & \multirow{-2}{*}{$\textsc{Sch}$} & \cellcolor[HTML]{EFEFEF}{\ul 0.365} & \cellcolor[HTML]{EFEFEF}0.428 & \cellcolor[HTML]{EFEFEF}0.451 & \cellcolor[HTML]{EFEFEF}0.465 & \cellcolor[HTML]{EFEFEF}0.476 & \cellcolor[HTML]{EFEFEF}0.480 & \cellcolor[HTML]{EFEFEF}0.488 & \cellcolor[HTML]{EFEFEF}0.492 & \cellcolor[HTML]{EFEFEF}0.496 & \cellcolor[HTML]{EFEFEF}0.498 & \cellcolor[HTML]{EFEFEF}0.518 & \cellcolor[HTML]{EFEFEF}0.531 & \cellcolor[HTML]{EFEFEF}0.537 & \cellcolor[HTML]{EFEFEF}0.542 & \cellcolor[HTML]{EFEFEF}0.547 & \cellcolor[HTML]{EFEFEF}0.549 & \cellcolor[HTML]{EFEFEF}0.551 & \cellcolor[HTML]{EFEFEF}0.546 \\ \midrule
& & {0.431} & {\ul 0.493} & {\ul 0.514} & {\ul 0.530} & {\ul 0.539} & {\ul 0.547} & {\ul 0.552} & {\ul 0.559} & {\ul 0.563} & {\ul 0.566} & \textbf{0.590} & \textbf{0.600} & \textbf{0.609} & \textbf{0.613} & {\ul 0.615} & \textbf{0.619} & \textbf{0.620} & \textbf{0.620} \\
& \multirow{-2}{*}{$\textsc{Gauss}$} & \cellcolor[HTML]{EFEFEF}0.362 & \cellcolor[HTML]{EFEFEF}{\ul 0.434} & \cellcolor[HTML]{EFEFEF}{\ul 0.455} & \cellcolor[HTML]{EFEFEF}{\ul 0.473} & \cellcolor[HTML]{EFEFEF}{\ul 0.482} & \cellcolor[HTML]{EFEFEF}\textbf{0.491} & \cellcolor[HTML]{EFEFEF}\textbf{0.497} & \cellcolor[HTML]{EFEFEF}\textbf{0.504} & \cellcolor[HTML]{EFEFEF}\textbf{0.508} & \cellcolor[HTML]{EFEFEF}\textbf{0.511} & \cellcolor[HTML]{EFEFEF}{\ul 0.537} & \cellcolor[HTML]{EFEFEF}{\ul 0.546} & \cellcolor[HTML]{EFEFEF}{\ul 0.555} & \cellcolor[HTML]{EFEFEF}{\ul 0.559} & \cellcolor[HTML]{EFEFEF}{\ul 0.562} & \cellcolor[HTML]{EFEFEF}{\ul \textbf{0.566}} & \cellcolor[HTML]{EFEFEF}\textbf{0.567} & \cellcolor[HTML]{EFEFEF}{0.565} \\
 & & \textbf{0.443} & \textbf{0.500} & \textbf{0.525} & \textbf{0.538} & \textbf{0.547} & \textbf{0.553} & \textbf{0.558} & \textbf{0.563} & \textbf{0.567} & \textbf{0.570} & {\ul 0.588} & 0.597 & 0.603 & 0.609 & 0.612 & 0.614 & 0.616 & 0.615 \\
\multirow{-4}{*}{\rotatebox{90}{\scriptsize \rotatebox{0}{\textsc{MKernelNE}}}} & \multirow{-2}{*}{$\textsc{Sch}$} & \cellcolor[HTML]{EFEFEF}{\ul 0.368} & \cellcolor[HTML]{EFEFEF}\textbf{0.433} & \cellcolor[HTML]{EFEFEF}\textbf{0.458} & \cellcolor[HTML]{EFEFEF}\textbf{0.475} & \cellcolor[HTML]{EFEFEF}\textbf{0.483} & \cellcolor[HTML]{EFEFEF}{\ul 0.490} & \cellcolor[HTML]{EFEFEF}{\ul 0.495} & \cellcolor[HTML]{EFEFEF}{\ul 0.500} & \cellcolor[HTML]{EFEFEF}{\ul 0.505} & \cellcolor[HTML]{EFEFEF}0.508 & \cellcolor[HTML]{EFEFEF}0.527 & \cellcolor[HTML]{EFEFEF}0.537 & \cellcolor[HTML]{EFEFEF}0.544 & \cellcolor[HTML]{EFEFEF}0.550 & \cellcolor[HTML]{EFEFEF}0.553 & \cellcolor[HTML]{EFEFEF}0.557 & {\cellcolor[HTML]{EFEFEF}0.557} & \cellcolor[HTML]{EFEFEF}0.556\\ \\ \bottomrule
\end{tabular}%
}
\end{table*}

\begin{table*}[]
\caption{Node classification task for varying training sizes on \textsl{Cora}. For each method, the rows show the Micro-$F_1$ and Macro-$F_1$ scores, respectively.}
\label{tab:classification_cora}
\centering
\resizebox{\textwidth}{!}{%
\begin{tabular}{crcccccccccccccccccc}\toprule
& & \textbf{1\%} & \textbf{2\%} & \textbf{3\%} & \textbf{4\%} & \textbf{5\%} & \textbf{6\%} & \textbf{7\%} & \textbf{8\%} & \textbf{9\%} & \textbf{10\%} & \textbf{20\%} & \textbf{30\%} & \textbf{40\%} & \textbf{50\%} & \textbf{60\%} & \textbf{70\%} & \textbf{80\%} & \textbf{90\%} \\\midrule
& & 0.522 & 0.617 & 0.660 & 0.688 & 0.703 & 0.715 & 0.724 & 0.732 & 0.742 & 0.747 & 0.782 & 0.799 & 0.808 & 0.815 & 0.821 & 0.825 & 0.827 & 0.832 \\
& \multirow{-2}{*}{$\textsc{DeepWalk}$} & \cellcolor[HTML]{EFEFEF}0.442 & \cellcolor[HTML]{EFEFEF}0.569 & \cellcolor[HTML]{EFEFEF}0.628 & \cellcolor[HTML]{EFEFEF}0.664 & \cellcolor[HTML]{EFEFEF}0.682 & \cellcolor[HTML]{EFEFEF}0.698 & \cellcolor[HTML]{EFEFEF}0.709 & \cellcolor[HTML]{EFEFEF}0.717 & \cellcolor[HTML]{EFEFEF}0.727 & \cellcolor[HTML]{EFEFEF}0.735 & \cellcolor[HTML]{EFEFEF}0.771 & \cellcolor[HTML]{EFEFEF}0.788 & \cellcolor[HTML]{EFEFEF}0.798 & \cellcolor[HTML]{EFEFEF}0.806 & \cellcolor[HTML]{EFEFEF}0.811 & \cellcolor[HTML]{EFEFEF}0.815 & \cellcolor[HTML]{EFEFEF}0.816 & \cellcolor[HTML]{EFEFEF}0.821 \\
& & 0.570 & 0.659 & 0.695 & 0.720 & 0.734 & 0.743 & 0.752 & 0.759 & 0.765 & 0.770 & 0.800 & 0.816 & 0.824 & 0.831 & 0.835 & 0.839 & 0.842 & 0.845 \\
& \multirow{-2}{*}{$\textsc{Node2Vec}$} & \cellcolor[HTML]{EFEFEF}0.489 & \cellcolor[HTML]{EFEFEF}0.612 & \cellcolor[HTML]{EFEFEF}0.662 & \cellcolor[HTML]{EFEFEF}0.694 & \cellcolor[HTML]{EFEFEF}0.714 & \cellcolor[HTML]{EFEFEF}0.724 & \cellcolor[HTML]{EFEFEF}0.735 & \cellcolor[HTML]{EFEFEF}0.743 & \cellcolor[HTML]{EFEFEF}0.751 & \cellcolor[HTML]{EFEFEF}0.755 & \cellcolor[HTML]{EFEFEF}0.788 & \cellcolor[HTML]{EFEFEF}0.804 & \cellcolor[HTML]{EFEFEF}0.813 & \cellcolor[HTML]{EFEFEF}0.820 & \cellcolor[HTML]{EFEFEF}0.824 & \cellcolor[HTML]{EFEFEF}0.828 & \cellcolor[HTML]{EFEFEF}0.830 & \cellcolor[HTML]{EFEFEF}0.832 \\
& & 0.351 & 0.416 & 0.463 & 0.498 & 0.521 & 0.546 & 0.566 & 0.581 & 0.598 & 0.609 & 0.673 & 0.701 & 0.719 & 0.728 & 0.735 & 0.741 & 0.743 & 0.747 \\
& \multirow{-2}{*}{$\textsc{LINE}$} & \cellcolor[HTML]{EFEFEF}0.223 & \cellcolor[HTML]{EFEFEF}0.306 & \cellcolor[HTML]{EFEFEF}0.373 & \cellcolor[HTML]{EFEFEF}0.425 & \cellcolor[HTML]{EFEFEF}0.457 & \cellcolor[HTML]{EFEFEF}0.492 & \cellcolor[HTML]{EFEFEF}0.519 & \cellcolor[HTML]{EFEFEF}0.543 & \cellcolor[HTML]{EFEFEF}0.565 & \cellcolor[HTML]{EFEFEF}0.578 & \cellcolor[HTML]{EFEFEF}0.659 & \cellcolor[HTML]{EFEFEF}0.691 & \cellcolor[HTML]{EFEFEF}0.710 & \cellcolor[HTML]{EFEFEF}0.720 & \cellcolor[HTML]{EFEFEF}0.727 & \cellcolor[HTML]{EFEFEF}0.733 & \cellcolor[HTML]{EFEFEF}0.736 & \cellcolor[HTML]{EFEFEF}0.738 \\
& & 0.278 & 0.284 & 0.297 & 0.301 & 0.302 & 0.301 & 0.302 & 0.302 & 0.302 & 0.302 & 0.303 & 0.302 & 0.302 & 0.302 & 0.303 & 0.304 & 0.303 & 0.306 \\
& \multirow{-2}{*}{$\textsc{HOPE}$} & \cellcolor[HTML]{EFEFEF}0.070 & \cellcolor[HTML]{EFEFEF}0.067 & \cellcolor[HTML]{EFEFEF}0.067 & \cellcolor[HTML]{EFEFEF}0.066 & \cellcolor[HTML]{EFEFEF}0.067 & \cellcolor[HTML]{EFEFEF}0.067 & \cellcolor[HTML]{EFEFEF}0.066 & \cellcolor[HTML]{EFEFEF}0.066 & \cellcolor[HTML]{EFEFEF}0.066 & \cellcolor[HTML]{EFEFEF}0.066 & \cellcolor[HTML]{EFEFEF}0.067 & \cellcolor[HTML]{EFEFEF}0.067 & \cellcolor[HTML]{EFEFEF}0.067 & \cellcolor[HTML]{EFEFEF}0.067 & \cellcolor[HTML]{EFEFEF}0.067 & \cellcolor[HTML]{EFEFEF}0.068 & \cellcolor[HTML]{EFEFEF}0.070 & \cellcolor[HTML]{EFEFEF}0.074 \\
 &  & {0.605} & {0.680} & {0.708} & {0.727} & {0.737} & {0.753} & {0.761} & {0.763} & {0.766} & {0.774} & {0.798} & {0.813} & {0.818} & {0.828} & {0.829} & {0.827} & {0.837} & {0.829} \\
 & \multirow{-2}{*}{$\textsc{VERSE}$} & \cellcolor[HTML]{EFEFEF}{0.527} & \cellcolor[HTML]{EFEFEF}{0.628} & \cellcolor[HTML]{EFEFEF}{0.675} & \cellcolor[HTML]{EFEFEF}{0.704} & \cellcolor[HTML]{EFEFEF}{0.717} & \cellcolor[HTML]{EFEFEF}{0.737} & \cellcolor[HTML]{EFEFEF}{0.746} & \cellcolor[HTML]{EFEFEF}{0.752} & \cellcolor[HTML]{EFEFEF}{0.754} & \cellcolor[HTML]{EFEFEF}{0.764} & \cellcolor[HTML]{EFEFEF}{0.788} & \cellcolor[HTML]{EFEFEF}{0.804} & \cellcolor[HTML]{EFEFEF}{0.809} & \cellcolor[HTML]{EFEFEF}{0.822} & \cellcolor[HTML]{EFEFEF}{0.821} & \cellcolor[HTML]{EFEFEF}{0.818} & \cellcolor[HTML]{EFEFEF}{0.833} & \cellcolor[HTML]{EFEFEF}{0.819} \\
 &  & {0.337} & {0.380} & {0.440} & {0.475} & {0.513} & {0.529} & {0.574} & {0.604} & {0.624} & {0.628} & {0.729} & {0.754} & {0.785} & {0.794} & {0.804} & {0.803} & {0.810} & {0.816} \\
 & \multirow{-2}{*}{$\textsc{ProNE}$} & \cellcolor[HTML]{EFEFEF}{0.202} & \cellcolor[HTML]{EFEFEF}{0.223} & \cellcolor[HTML]{EFEFEF}{0.319} & \cellcolor[HTML]{EFEFEF}{0.363} & \cellcolor[HTML]{EFEFEF}{0.434} & \cellcolor[HTML]{EFEFEF}{0.446} & \cellcolor[HTML]{EFEFEF}{0.520} & \cellcolor[HTML]{EFEFEF}{0.559} & \cellcolor[HTML]{EFEFEF}{0.579} & \cellcolor[HTML]{EFEFEF}{0.589} & \cellcolor[HTML]{EFEFEF}{0.712} & \cellcolor[HTML]{EFEFEF}{0.743} & \cellcolor[HTML]{EFEFEF}{0.776} & \cellcolor[HTML]{EFEFEF}{0.784} & \cellcolor[HTML]{EFEFEF}{0.793} & \cellcolor[HTML]{EFEFEF}{0.793} & \cellcolor[HTML]{EFEFEF}{0.801} & \cellcolor[HTML]{EFEFEF}{0.809}
\\
& & 0.534 & 0.636 & 0.693 & 0.716 & 0.735 & 0.748 & 0.757 & 0.767 & 0.770 & 0.773 & \textbf{0.807} & \textbf{0.821} & {\ul 0.828} & {\ul 0.834} & {\ul 0.839} & {\ul0.841} & {\ul 0.839} & {\ul 0.844} \\
 & \multirow{-2}{*}{$\textsc{NetMF}$} & \cellcolor[HTML]{EFEFEF}0.461 & \cellcolor[HTML]{EFEFEF}0.591 & \cellcolor[HTML]{EFEFEF}0.667 & \cellcolor[HTML]{EFEFEF}0.694 & \cellcolor[HTML]{EFEFEF}0.717 & \cellcolor[HTML]{EFEFEF}0.731 & \cellcolor[HTML]{EFEFEF}0.741 & \cellcolor[HTML]{EFEFEF}0.751 & \cellcolor[HTML]{EFEFEF}0.757 & \cellcolor[HTML]{EFEFEF}0.760 & \cellcolor[HTML]{EFEFEF}\textbf{0.797} & \cellcolor[HTML]{EFEFEF}\textbf{0.811} & \cellcolor[HTML]{EFEFEF}\textbf{0.819} & \cellcolor[HTML]{EFEFEF}{\ul 0.824} & \cellcolor[HTML]{EFEFEF}{\ul 0.830} & \cellcolor[HTML]{EFEFEF}{\ul 0.832} & \cellcolor[HTML]{EFEFEF}{\ul 0.831} & \cellcolor[HTML]{EFEFEF}{\ul 0.835}\\
 &  & 0.397 & 0.470 & 0.497 & 0.530 & 0.551 & 0.568 & 0.578 & 0.587 & 0.596 & 0.601 & 0.643 & 0.674 & 0.698 & 0.714 & 0.728 & 0.735 & 0.741 & 0.744 \\
 & \multirow{-2}{*}{$\textsc{GEMSEC}$} & \cellcolor[HTML]{EFEFEF}0.317 & \cellcolor[HTML]{EFEFEF}0.406 & \cellcolor[HTML]{EFEFEF}0.439 & \cellcolor[HTML]{EFEFEF}0.477 & \cellcolor[HTML]{EFEFEF}0.504 & \cellcolor[HTML]{EFEFEF}0.527 & \cellcolor[HTML]{EFEFEF}0.535 & \cellcolor[HTML]{EFEFEF}0.546 & \cellcolor[HTML]{EFEFEF}0.556 & \cellcolor[HTML]{EFEFEF}0.562 & \cellcolor[HTML]{EFEFEF}0.611 & \cellcolor[HTML]{EFEFEF}0.646 & \cellcolor[HTML]{EFEFEF}0.672 & \cellcolor[HTML]{EFEFEF}0.689 & \cellcolor[HTML]{EFEFEF}0.704 & \cellcolor[HTML]{EFEFEF}0.713 & \cellcolor[HTML]{EFEFEF}0.719 & \cellcolor[HTML]{EFEFEF}0.722\\
 &  & 0.419 & 0.507 & 0.549 & 0.580 & 0.604 & 0.622 & 0.633 & 0.642 & 0.652 & 0.656 & 0.700 & 0.717 & 0.725 & 0.732 & 0.736 & 0.736 & 0.744 & 0.742 \\
 \multirow{-18}{*}{\rotatebox{90}{Baselines}} & \multirow{-2}{*}{$\textsc{M-NMF}$} & \cellcolor[HTML]{EFEFEF}0.354 & \cellcolor[HTML]{EFEFEF}0.459 & \cellcolor[HTML]{EFEFEF}0.507 & \cellcolor[HTML]{EFEFEF}0.550 & \cellcolor[HTML]{EFEFEF}0.575 & \cellcolor[HTML]{EFEFEF}0.598 & \cellcolor[HTML]{EFEFEF}0.609 & \cellcolor[HTML]{EFEFEF}0.622 & \cellcolor[HTML]{EFEFEF}0.632 & \cellcolor[HTML]{EFEFEF}0.638 & \cellcolor[HTML]{EFEFEF}0.687 & \cellcolor[HTML]{EFEFEF}0.706 & \cellcolor[HTML]{EFEFEF}0.716 & \cellcolor[HTML]{EFEFEF}0.722 & \cellcolor[HTML]{EFEFEF}0.728 & \cellcolor[HTML]{EFEFEF}0.728 & \cellcolor[HTML]{EFEFEF}0.735 & \cellcolor[HTML]{EFEFEF}0.734
\\\midrule
& & {\ul 0.628} & 0.696 & 0.721 & 0.739 & 0.748 & {\ul 0.759} & {\ul 0.766} & {\ul 0.772} & {\ul 0.775} & {\ul 0.780} & {\ul 0.806} & {\ul 0.820} & \textbf{0.829} & \textbf{0.837} & \textbf{0.843} & \textbf{0.846} & \textbf{0.849} & \textbf{0.851} \\
& \multirow{-2}{*}{$\textsc{Gauss}$} & \cellcolor[HTML]{EFEFEF}\textbf{0.566} & \cellcolor[HTML]{EFEFEF}\textbf{0.664} & \cellcolor[HTML]{EFEFEF}\textbf{0.696} & \cellcolor[HTML]{EFEFEF}{\ul 0.721} & \cellcolor[HTML]{EFEFEF}{\ul 0.730} & \cellcolor[HTML]{EFEFEF}{\ul 0.743} & \cellcolor[HTML]{EFEFEF}{\ul 0.752} & \cellcolor[HTML]{EFEFEF}{\ul 0.758} & \cellcolor[HTML]{EFEFEF}{\ul 0.762} & \cellcolor[HTML]{EFEFEF}{\ul 0.767} & \cellcolor[HTML]{EFEFEF}{\ul 0.794} & \cellcolor[HTML]{EFEFEF}{\ul 0.809} & \cellcolor[HTML]{EFEFEF}{\ul 0.818} & \cellcolor[HTML]{EFEFEF}\textbf{0.826} & \cellcolor[HTML]{EFEFEF}\textbf{0.832} & \cellcolor[HTML]{EFEFEF}\textbf{0.836} & \cellcolor[HTML]{EFEFEF}\textbf{0.838} & \cellcolor[HTML]{EFEFEF}\textbf{0.840} \\
& & 0.619 & 0.695 & 0.722 & 0.736 & 0.745 & 0.750 & 0.755 & 0.759 & 0.763 & 0.765 & 0.783 & 0.790 & 0.796 & 0.800 & 0.803 & 0.806 & 0.807 & 0.812 \\
\multirow{-4}{*}{\rotatebox{90}{\scriptsize \textsc{KernelNE}}} & \multirow{-2}{*}{$\textsc{Sch}$} & \cellcolor[HTML]{EFEFEF}0.521 & \cellcolor[HTML]{EFEFEF}0.631 & \cellcolor[HTML]{EFEFEF}0.671 & \cellcolor[HTML]{EFEFEF}0.697 & \cellcolor[HTML]{EFEFEF}0.712 & \cellcolor[HTML]{EFEFEF}0.721 & \cellcolor[HTML]{EFEFEF}0.728 & \cellcolor[HTML]{EFEFEF}0.734 & \cellcolor[HTML]{EFEFEF}0.743 & \cellcolor[HTML]{EFEFEF}0.745 & \cellcolor[HTML]{EFEFEF}0.771 & \cellcolor[HTML]{EFEFEF}0.780 & \cellcolor[HTML]{EFEFEF}0.786 & \cellcolor[HTML]{EFEFEF}0.790 & \cellcolor[HTML]{EFEFEF}0.792 & \cellcolor[HTML]{EFEFEF}0.795 & \cellcolor[HTML]{EFEFEF}0.795 & \cellcolor[HTML]{EFEFEF}0.799 \\\midrule
& & \textbf{0.631} & \textbf{0.701} & \textbf{0.731} & \textbf{0.748} & \textbf{0.757} & \textbf{0.764} & \textbf{0.770} & \textbf{0.775} & \textbf{0.778} & \textbf{0.781} & 0.801 & 0.812 & 0.819 & 0.823 & 0.827 & 0.828 & 0.833 & 0.833 \\
& \multirow{-2}{*}{$\textsc{Gauss}$} & \cellcolor[HTML]{EFEFEF}{\ul 0.562} & \cellcolor[HTML]{EFEFEF}{\ul 0.656} & \cellcolor[HTML]{EFEFEF}\textbf{0.696} & \cellcolor[HTML]{EFEFEF}\textbf{0.723} & \cellcolor[HTML]{EFEFEF}\textbf{0.736} & \cellcolor[HTML]{EFEFEF}\textbf{0.746} & \cellcolor[HTML]{EFEFEF}\textbf{0.755} & \cellcolor[HTML]{EFEFEF}\textbf{0.761} & \cellcolor[HTML]{EFEFEF}\textbf{0.765} & \cellcolor[HTML]{EFEFEF}\textbf{0.769} & \cellcolor[HTML]{EFEFEF}0.791 & \cellcolor[HTML]{EFEFEF}0.801 & \cellcolor[HTML]{EFEFEF}0.808 & \cellcolor[HTML]{EFEFEF}0.813 & \cellcolor[HTML]{EFEFEF}0.817 & \cellcolor[HTML]{EFEFEF}0.818 & \cellcolor[HTML]{EFEFEF}0.822 & \cellcolor[HTML]{EFEFEF}0.820 \\
& & 0.623 & {\ul 0.699} & {\ul 0.728} & {\ul 0.742} & {\ul 0.751} & 0.757 & 0.762 & 0.767 & 0.770 & 0.772 & 0.788 & 0.797 & 0.803 & 0.806 & 0.812 & 0.812 & 0.817 & 0.818 \\
\multirow{-4}{*}{\rotatebox{90}{\scriptsize \textsc{MKernelNE}}} & \multirow{-2}{*}{$\textsc{Sch}$} & \cellcolor[HTML]{EFEFEF}0.527 & \cellcolor[HTML]{EFEFEF}0.637 & \cellcolor[HTML]{EFEFEF}0.686 & \cellcolor[HTML]{EFEFEF}0.708 & \cellcolor[HTML]{EFEFEF}0.723 & \cellcolor[HTML]{EFEFEF}0.733 & \cellcolor[HTML]{EFEFEF}0.741 & \cellcolor[HTML]{EFEFEF}0.749 & \cellcolor[HTML]{EFEFEF}0.753 & \cellcolor[HTML]{EFEFEF}0.756 & \cellcolor[HTML]{EFEFEF}0.777 & \cellcolor[HTML]{EFEFEF}0.787 & \cellcolor[HTML]{EFEFEF}0.793 & \cellcolor[HTML]{EFEFEF}0.796 & \cellcolor[HTML]{EFEFEF}0.801 & \cellcolor[HTML]{EFEFEF}0.802 & \cellcolor[HTML]{EFEFEF}0.805 & \cellcolor[HTML]{EFEFEF}0.804\\ \\\bottomrule
\end{tabular}%
}
\end{table*}

\begin{table*}[]
\caption{Node classification task for varying training sizes on \textsl{Dblp}. For each method, the rows show the Micro-$F_1$ and Macro-$F_1$ scores, respectively.}
\label{tab:classification_dblp}
\centering
\resizebox{\textwidth}{!}{%
\begin{tabular}{crcccccccccccccccccc}\toprule
& & \textbf{1\%} & \textbf{2\%} & \textbf{3\%} & \textbf{4\%} & \textbf{5\%} & \textbf{6\%} & \textbf{7\%} & \textbf{8\%} & \textbf{9\%} & \textbf{10\%} & \textbf{20\%} & \textbf{30\%} & \textbf{40\%} & \textbf{50\%} & \textbf{60\%} & \textbf{70\%} & \textbf{80\%} & \textbf{90\%} \\\midrule
& & 0.518 & 0.550 & 0.573 & 0.588 & 0.597 & 0.603 & 0.607 & 0.611 & 0.614 & 0.616 & 0.627 & 0.630 & 0.632 & 0.633 & 0.634 & 0.635 & 0.635 & 0.636 \\
& \multirow{-2}{*}{$\textsc{DeepWalk}$} & \cellcolor[HTML]{EFEFEF}0.464 & \cellcolor[HTML]{EFEFEF}0.496 & \cellcolor[HTML]{EFEFEF}0.515 & \cellcolor[HTML]{EFEFEF}0.527 & \cellcolor[HTML]{EFEFEF}0.535 & \cellcolor[HTML]{EFEFEF}0.540 & \cellcolor[HTML]{EFEFEF}0.543 & \cellcolor[HTML]{EFEFEF}0.547 & \cellcolor[HTML]{EFEFEF}0.550 & \cellcolor[HTML]{EFEFEF}0.551 & \cellcolor[HTML]{EFEFEF}0.560 & \cellcolor[HTML]{EFEFEF}0.563 & \cellcolor[HTML]{EFEFEF}0.564 & \cellcolor[HTML]{EFEFEF}0.565 & \cellcolor[HTML]{EFEFEF}0.566 & \cellcolor[HTML]{EFEFEF}0.566 & \cellcolor[HTML]{EFEFEF}0.567 & \cellcolor[HTML]{EFEFEF}0.567 \\
& & 0.557 & 0.575 & 0.590 & 0.599 & 0.606 & 0.612 & 0.615 & 0.618 & 0.621 & 0.623 & 0.633 & 0.636 & 0.638 & 0.639 & 0.640 & 0.641 & 0.641 & 0.641 \\
& \multirow{-2}{*}{$\textsc{Node2Vec}$} & \cellcolor[HTML]{EFEFEF}0.497 & \cellcolor[HTML]{EFEFEF}0.517 & \cellcolor[HTML]{EFEFEF}0.531 & \cellcolor[HTML]{EFEFEF}0.541 & \cellcolor[HTML]{EFEFEF}0.546 & \cellcolor[HTML]{EFEFEF}0.551 & \cellcolor[HTML]{EFEFEF}0.554 & \cellcolor[HTML]{EFEFEF}0.557 & \cellcolor[HTML]{EFEFEF}0.559 & \cellcolor[HTML]{EFEFEF}0.561 & \cellcolor[HTML]{EFEFEF}0.568 & \cellcolor[HTML]{EFEFEF}0.571 & \cellcolor[HTML]{EFEFEF}0.573 & \cellcolor[HTML]{EFEFEF}0.574 & \cellcolor[HTML]{EFEFEF}0.574 & \cellcolor[HTML]{EFEFEF}0.575 & \cellcolor[HTML]{EFEFEF}0.574 & \cellcolor[HTML]{EFEFEF}0.574 \\
& & 0.525 & 0.553 & 0.567 & 0.576 & 0.581 & 0.585 & 0.588 & 0.590 & 0.593 & 0.594 & 0.603 & 0.606 & 0.608 & 0.610 & 0.610 & 0.611 & 0.611 & 0.611 \\
& \multirow{-2}{*}{$\textsc{LINE}$} & \cellcolor[HTML]{EFEFEF}0.430 & \cellcolor[HTML]{EFEFEF}0.469 & \cellcolor[HTML]{EFEFEF}0.488 & \cellcolor[HTML]{EFEFEF}0.498 & \cellcolor[HTML]{EFEFEF}0.505 & \cellcolor[HTML]{EFEFEF}0.510 & \cellcolor[HTML]{EFEFEF}0.514 & \cellcolor[HTML]{EFEFEF}0.517 & \cellcolor[HTML]{EFEFEF}0.520 & \cellcolor[HTML]{EFEFEF}0.522 & \cellcolor[HTML]{EFEFEF}0.532 & \cellcolor[HTML]{EFEFEF}0.536 & \cellcolor[HTML]{EFEFEF}0.538 & \cellcolor[HTML]{EFEFEF}0.540 & \cellcolor[HTML]{EFEFEF}0.540 & \cellcolor[HTML]{EFEFEF}0.541 & \cellcolor[HTML]{EFEFEF}0.541 & \cellcolor[HTML]{EFEFEF}0.541 \\
& & 0.377 & 0.379 & 0.379 & 0.379 & 0.379 & 0.379 & 0.379 & 0.379 & 0.379 & 0.379 & 0.379 & 0.380 & 0.381 & 0.383 & 0.385 & 0.387 & 0.388 & 0.391 \\
& \multirow{-2}{*}{$\textsc{HOPE}$} & \cellcolor[HTML]{EFEFEF}0.137 & \cellcolor[HTML]{EFEFEF}0.137 & \cellcolor[HTML]{EFEFEF}0.137 & \cellcolor[HTML]{EFEFEF}0.137 & \cellcolor[HTML]{EFEFEF}0.137 & \cellcolor[HTML]{EFEFEF}0.137 & \cellcolor[HTML]{EFEFEF}0.137 & \cellcolor[HTML]{EFEFEF}0.137 & \cellcolor[HTML]{EFEFEF}0.137 & \cellcolor[HTML]{EFEFEF}0.138 & \cellcolor[HTML]{EFEFEF}0.139 & \cellcolor[HTML]{EFEFEF}0.140 & \cellcolor[HTML]{EFEFEF}0.143 & \cellcolor[HTML]{EFEFEF}0.146 & \cellcolor[HTML]{EFEFEF}0.149 & \cellcolor[HTML]{EFEFEF}0.152 & \cellcolor[HTML]{EFEFEF}0.156 & \cellcolor[HTML]{EFEFEF}0.160 \\
 &  & {0.550} & {0.566} & {0.585} & {0.593} & {0.599} & {0.605} & {0.610} & {0.613} & {0.615} & {0.617} & {0.629} & {0.629} & {0.632} & {0.635} & {0.633} & {0.632} & {0.634} & {0.636} \\
 & \multirow{-2}{*}{$\textsc{VERSE}$} & \cellcolor[HTML]{EFEFEF}{0.492} & \cellcolor[HTML]{EFEFEF}{0.508} & \cellcolor[HTML]{EFEFEF}{0.528} & \cellcolor[HTML]{EFEFEF}{0.534} & \cellcolor[HTML]{EFEFEF}{0.539} & \cellcolor[HTML]{EFEFEF}{0.545} & \cellcolor[HTML]{EFEFEF}{0.552} & \cellcolor[HTML]{EFEFEF}{0.551} & \cellcolor[HTML]{EFEFEF}{0.554} & \cellcolor[HTML]{EFEFEF}{0.555} & \cellcolor[HTML]{EFEFEF}{0.565} & \cellcolor[HTML]{EFEFEF}{0.566} & \cellcolor[HTML]{EFEFEF}{0.568} & \cellcolor[HTML]{EFEFEF}{0.571} & \cellcolor[HTML]{EFEFEF}{0.567} & \cellcolor[HTML]{EFEFEF}{0.566} & \cellcolor[HTML]{EFEFEF}{0.570} & \cellcolor[HTML]{EFEFEF}{0.571} \\
 &  & {0.475} & {0.525} & {0.550} & {0.564} & {0.574} & {0.581} & {0.589} & {0.593} & {0.596} & {0.600} & {0.615} & {0.624} & {0.626} & {0.628} & {0.629} & {0.632} & {0.632} & {0.634} \\
 & \multirow{-2}{*}{$\textsc{ProNE}$} & \cellcolor[HTML]{EFEFEF}{0.363} & \cellcolor[HTML]{EFEFEF}{0.436} & \cellcolor[HTML]{EFEFEF}{0.457} & \cellcolor[HTML]{EFEFEF}{0.478} & \cellcolor[HTML]{EFEFEF}{0.493} & \cellcolor[HTML]{EFEFEF}{0.498} & \cellcolor[HTML]{EFEFEF}{0.510} & \cellcolor[HTML]{EFEFEF}{0.517} & \cellcolor[HTML]{EFEFEF}{0.520} & \cellcolor[HTML]{EFEFEF}{0.520} & \cellcolor[HTML]{EFEFEF}{0.543} & \cellcolor[HTML]{EFEFEF}{0.553} & \cellcolor[HTML]{EFEFEF}{0.554} & \cellcolor[HTML]{EFEFEF}{0.557} & \cellcolor[HTML]{EFEFEF}{0.559} & \cellcolor[HTML]{EFEFEF}{0.562} & \cellcolor[HTML]{EFEFEF}{0.563} & \cellcolor[HTML]{EFEFEF}{0.564}
\\
& & 0.564 & 0.577 & 0.586 & 0.589 & 0.593 & 0.596 & 0.599 & 0.601 & 0.604 & 0.605 & 0.613 & 0.617 & 0.619 & 0.620 & 0.620 & 0.623 & 0.623 & 0.623 \\
& \multirow{-2}{*}{$\textsc{NetMF}$} & \cellcolor[HTML]{EFEFEF}0.463 & \cellcolor[HTML]{EFEFEF}0.490 & \cellcolor[HTML]{EFEFEF}0.503 & \cellcolor[HTML]{EFEFEF}0.506 & \cellcolor[HTML]{EFEFEF}0.510 & \cellcolor[HTML]{EFEFEF}0.513 & \cellcolor[HTML]{EFEFEF}0.517 & \cellcolor[HTML]{EFEFEF}0.518 & \cellcolor[HTML]{EFEFEF}0.521 & \cellcolor[HTML]{EFEFEF}0.522 & \cellcolor[HTML]{EFEFEF}0.528 & \cellcolor[HTML]{EFEFEF}0.530 & \cellcolor[HTML]{EFEFEF}0.532 & \cellcolor[HTML]{EFEFEF}0.531 & \cellcolor[HTML]{EFEFEF}0.531 & \cellcolor[HTML]{EFEFEF}0.533 & \cellcolor[HTML]{EFEFEF}0.533 & \cellcolor[HTML]{EFEFEF}0.533\\
 &  & 0.511 & 0.538 & 0.555 & 0.566 & 0.575 & 0.583 & 0.588 & 0.591 & 0.593 & 0.597 & 0.607 & 0.611 & 0.613 & 0.614 & 0.615 & 0.615 & 0.616 & 0.615 \\
  & \multirow{-2}{*}{$\textsc{GEMSEC}$} & \cellcolor[HTML]{EFEFEF}0.453 & \cellcolor[HTML]{EFEFEF}0.477 & \cellcolor[HTML]{EFEFEF}0.492 & \cellcolor[HTML]{EFEFEF}0.501 & \cellcolor[HTML]{EFEFEF}0.508 & \cellcolor[HTML]{EFEFEF}0.513 & \cellcolor[HTML]{EFEFEF}0.518 & \cellcolor[HTML]{EFEFEF}0.519 & \cellcolor[HTML]{EFEFEF}0.521 & \cellcolor[HTML]{EFEFEF}0.523 & \cellcolor[HTML]{EFEFEF}0.531 & \cellcolor[HTML]{EFEFEF}0.534 & \cellcolor[HTML]{EFEFEF}0.534 & \cellcolor[HTML]{EFEFEF}0.535 & \cellcolor[HTML]{EFEFEF}0.535 & \cellcolor[HTML]{EFEFEF}0.537 & \cellcolor[HTML]{EFEFEF}0.537 & \cellcolor[HTML]{EFEFEF}0.536\\
 &  & 0.473 & 0.501 & 0.518 & 0.527 & 0.533 & 0.540 & 0.542 & 0.545 & 0.548 & 0.551 & 0.563 & 0.568 & 0.571 & 0.574 & 0.574 & 0.577 & 0.577 & 0.579 \\
\multirow{-18}{*}{\rotatebox{90}{Baselines}} & \multirow{-2}{*}{$\textsc{M-NMF}$} & \cellcolor[HTML]{EFEFEF}0.302 & \cellcolor[HTML]{EFEFEF}0.345 & \cellcolor[HTML]{EFEFEF}0.369 & \cellcolor[HTML]{EFEFEF}0.383 & \cellcolor[HTML]{EFEFEF}0.392 & \cellcolor[HTML]{EFEFEF}0.400 & \cellcolor[HTML]{EFEFEF}0.406 & \cellcolor[HTML]{EFEFEF}0.410 & \cellcolor[HTML]{EFEFEF}0.414 & \cellcolor[HTML]{EFEFEF}0.418 & \cellcolor[HTML]{EFEFEF}0.436 & \cellcolor[HTML]{EFEFEF}0.443 & \cellcolor[HTML]{EFEFEF}0.447 & \cellcolor[HTML]{EFEFEF}0.450 & \cellcolor[HTML]{EFEFEF}0.451 & \cellcolor[HTML]{EFEFEF}0.454 & \cellcolor[HTML]{EFEFEF}0.455 & \cellcolor[HTML]{EFEFEF}0.455
\\\midrule
& & 0.558 & 0.575 & 0.587 & 0.595 & 0.601 & 0.606 & 0.610 & 0.613 & 0.615 & 0.617 & 0.627 & 0.630 & 0.632 & 0.633 & 0.634 & 0.635 & 0.636 & 0.636 \\
& \multirow{-2}{*}{$\textsc{Gauss}$} & \cellcolor[HTML]{EFEFEF}0.498 & \cellcolor[HTML]{EFEFEF}0.517 & \cellcolor[HTML]{EFEFEF}0.528 & \cellcolor[HTML]{EFEFEF}0.536 & \cellcolor[HTML]{EFEFEF}0.542 & \cellcolor[HTML]{EFEFEF}0.545 & \cellcolor[HTML]{EFEFEF}0.548 & \cellcolor[HTML]{EFEFEF}0.551 & \cellcolor[HTML]{EFEFEF}0.552 & \cellcolor[HTML]{EFEFEF}0.554 & \cellcolor[HTML]{EFEFEF}0.562 & \cellcolor[HTML]{EFEFEF}0.564 & \cellcolor[HTML]{EFEFEF}0.566 & \cellcolor[HTML]{EFEFEF}0.566 & \cellcolor[HTML]{EFEFEF}0.567 & \cellcolor[HTML]{EFEFEF}0.567 & \cellcolor[HTML]{EFEFEF}0.568 & \cellcolor[HTML]{EFEFEF}0.569 \\
& & 0.599 & 0.609 & 0.613 & 0.617 & 0.620 & 0.621 & 0.623 & 0.624 & 0.625 & 0.627 & 0.632 & 0.635 & 0.636 & 0.637 & 0.638 & 0.638 & 0.638 & 0.640 \\
\multirow{-4}{*}{\rotatebox{90}{\scriptsize \textsc{KernelNE}}} & \multirow{-2}{*}{$\textsc{Sch}$} & \cellcolor[HTML]{EFEFEF}0.512 & \cellcolor[HTML]{EFEFEF}0.531 & \cellcolor[HTML]{EFEFEF}0.539 & \cellcolor[HTML]{EFEFEF}0.546 & \cellcolor[HTML]{EFEFEF}0.549 & \cellcolor[HTML]{EFEFEF}0.552 & \cellcolor[HTML]{EFEFEF}0.554 & \cellcolor[HTML]{EFEFEF}0.556 & \cellcolor[HTML]{EFEFEF}0.557 & \cellcolor[HTML]{EFEFEF}0.559 & \cellcolor[HTML]{EFEFEF}0.565 & \cellcolor[HTML]{EFEFEF}0.568 & \cellcolor[HTML]{EFEFEF}0.569 & \cellcolor[HTML]{EFEFEF}0.570 & \cellcolor[HTML]{EFEFEF}0.571 & \cellcolor[HTML]{EFEFEF}0.571 & \cellcolor[HTML]{EFEFEF}0.571 & \cellcolor[HTML]{EFEFEF}0.572 \\\midrule
& & {\ul 0.603} & {\ul 0.614} & \textbf{0.620} & \textbf{0.624} & \textbf{0.627} & \textbf{0.630} & \textbf{0.631} & \textbf{0.633} & \textbf{0.635} & \textbf{0.636} & \textbf{0.643} & \textbf{0.646} & \textbf{0.647} & \textbf{0.648} & \textbf{0.649} & \textbf{0.650} & \textbf{0.649} & \textbf{0.650} \\
& \multirow{-2}{*}{$\textsc{Gauss}$} & \cellcolor[HTML]{EFEFEF}\textbf{0.528} & \cellcolor[HTML]{EFEFEF}\textbf{0.547} & \cellcolor[HTML]{EFEFEF}\textbf{0.554} & \cellcolor[HTML]{EFEFEF}\textbf{0.560} & \cellcolor[HTML]{EFEFEF}\textbf{0.564} & \cellcolor[HTML]{EFEFEF}\textbf{0.566} & \cellcolor[HTML]{EFEFEF}\textbf{0.568} & \cellcolor[HTML]{EFEFEF}\textbf{0.570} & \cellcolor[HTML]{EFEFEF}\textbf{0.572} & \cellcolor[HTML]{EFEFEF}\textbf{0.573} & \cellcolor[HTML]{EFEFEF}\textbf{0.579} & \cellcolor[HTML]{EFEFEF}\textbf{0.582} & \cellcolor[HTML]{EFEFEF}\textbf{0.582} & \cellcolor[HTML]{EFEFEF}\textbf{0.583} & \cellcolor[HTML]{EFEFEF}\textbf{0.584} & \cellcolor[HTML]{EFEFEF}\textbf{0.584} & \cellcolor[HTML]{EFEFEF}\textbf{0.584} & \cellcolor[HTML]{EFEFEF}\textbf{0.584} \\
& & \textbf{0.607} & \textbf{0.615} & {\ul 0.619} & {\ul 0.621} & {\ul 0.624} & {\ul 0.625} & {\ul 0.627} & {\ul 0.628} & {\ul 0.630} & {\ul 0.631} & {\ul 0.637} & {\ul 0.639} & {\ul 0.641} & {\ul 0.641} & {\ul 0.642} & {\ul 0.643} & {\ul 0.643} & {\ul 0.643} \\
\multirow{-4}{*}{\rotatebox{90}{\scriptsize \textsc{MKernelNE}}} & \multirow{-2}{*}{$\textsc{Sch}$} & \cellcolor[HTML]{EFEFEF}{\ul 0.517} & \cellcolor[HTML]{EFEFEF}{\ul 0.538} & \cellcolor[HTML]{EFEFEF}{\ul 0.546} & \cellcolor[HTML]{EFEFEF}{\ul 0.551} & \cellcolor[HTML]{EFEFEF}{\ul 0.555} & \cellcolor[HTML]{EFEFEF}{\ul 0.557} & \cellcolor[HTML]{EFEFEF}{\ul 0.560} & \cellcolor[HTML]{EFEFEF}{\ul 0.561} & \cellcolor[HTML]{EFEFEF}{\ul 0.563} & \cellcolor[HTML]{EFEFEF}{\ul 0.564} & \cellcolor[HTML]{EFEFEF}{\ul 0.570} & \cellcolor[HTML]{EFEFEF}{\ul 0.572} & \cellcolor[HTML]{EFEFEF}{\ul 0.574} & \cellcolor[HTML]{EFEFEF}{\ul 0.574} & \cellcolor[HTML]{EFEFEF}{\ul 0.575} & \cellcolor[HTML]{EFEFEF}{\ul 0.576} & \cellcolor[HTML]{EFEFEF}{\ul 0.576} & \cellcolor[HTML]{EFEFEF}{\ul 0.576} \\\\\bottomrule
\end{tabular}%
}
\end{table*}

\begin{table*}[]
\centering
\caption{Node classification task for varying training sizes on \textsl{PPI}. For each method, the rows show the Micro-$F_1$ and Macro-$F_1$ scores, respectively.}
\label{tab:classification_ppi}
\resizebox{\textwidth}{!}{%
\begin{tabular}{crcccccccccccccccccc}\toprule
 &  & \textbf{1\%} & \textbf{2\%} & \textbf{3\%} & \textbf{4\%} & \textbf{5\%} & \textbf{6\%} & \textbf{7\%} & \textbf{8\%} & \textbf{9\%} & \textbf{10\%} & \textbf{20\%} & \textbf{30\%} & \textbf{40\%} & \textbf{50\%} & \textbf{60\%} & \textbf{70\%} & \textbf{80\%} & \textbf{90\%} \\\midrule
 &  & 0.093 & 0.110 & 0.122 & 0.131 & 0.137 & 0.143 & 0.147 & 0.151 & 0.155 & 0.156 & 0.176 & 0.188 & 0.197 & 0.206 & 0.212 & 0.218 & 0.223 & 0.226 \\
 & \multirow{-2}{*}{$\textsc{DeepWalk}$} & \cellcolor[HTML]{EFEFEF}0.056 & \cellcolor[HTML]{EFEFEF}{\ul 0.076} & \cellcolor[HTML]{EFEFEF}{\ul 0.090} & \cellcolor[HTML]{EFEFEF}{\ul 0.099} & \cellcolor[HTML]{EFEFEF}{\ul 0.107} & \cellcolor[HTML]{EFEFEF}{\ul 0.113} & \cellcolor[HTML]{EFEFEF}{0.118} & \cellcolor[HTML]{EFEFEF}{0.123} & \cellcolor[HTML]{EFEFEF}{0.126} & \cellcolor[HTML]{EFEFEF}0.129 & \cellcolor[HTML]{EFEFEF}0.151 & \cellcolor[HTML]{EFEFEF}0.164 & \cellcolor[HTML]{EFEFEF}0.173 & \cellcolor[HTML]{EFEFEF}0.181 & \cellcolor[HTML]{EFEFEF}0.186 & \cellcolor[HTML]{EFEFEF}0.190 & \cellcolor[HTML]{EFEFEF}0.193 & \cellcolor[HTML]{EFEFEF}0.192 \\
 &  & 0.097 & {\ul 0.115} & 0.128 & {\ul 0.136} & 0.143 & 0.148 & 0.154 & 0.157 & 0.161 & 0.164 & 0.185 & 0.196 & 0.205 & 0.211 & 0.216 & 0.221 & 0.224 & 0.226 \\
 & \multirow{-2}{*}{$\textsc{Node2Vec}$} & \cellcolor[HTML]{EFEFEF}\textbf{0.058} & \cellcolor[HTML]{EFEFEF}\textbf{0.078} & \cellcolor[HTML]{EFEFEF}\textbf{0.092} & \cellcolor[HTML]{EFEFEF}\textbf{0.101} & \cellcolor[HTML]{EFEFEF}\textbf{0.108} & \cellcolor[HTML]{EFEFEF}\textbf{0.116} & \cellcolor[HTML]{EFEFEF}\textbf{0.122} & \cellcolor[HTML]{EFEFEF}\textbf{0.125} & \cellcolor[HTML]{EFEFEF}{\ul 0.129} & \cellcolor[HTML]{EFEFEF}\textbf{0.132} & \cellcolor[HTML]{EFEFEF}0.155 & \cellcolor[HTML]{EFEFEF}0.167 & \cellcolor[HTML]{EFEFEF}0.175 & \cellcolor[HTML]{EFEFEF}0.180 & \cellcolor[HTML]{EFEFEF}0.185 & \cellcolor[HTML]{EFEFEF}0.188 & \cellcolor[HTML]{EFEFEF}0.189 & \cellcolor[HTML]{EFEFEF}0.190 \\
 &  & 0.091 & 0.109 & 0.123 & 0.133 & 0.142 & 0.149 & 0.156 & 0.162 & 0.166 & 0.170 & 0.197 & 0.210 & 0.219 & 0.226 & 0.231 & 0.235 & 0.238 & 0.242 \\
 & \multirow{-2}{*}{$\textsc{LINE}$} & \cellcolor[HTML]{EFEFEF}0.046 & \cellcolor[HTML]{EFEFEF}0.063 & \cellcolor[HTML]{EFEFEF}0.075 & \cellcolor[HTML]{EFEFEF}0.084 & \cellcolor[HTML]{EFEFEF}0.093 & \cellcolor[HTML]{EFEFEF}0.099 & \cellcolor[HTML]{EFEFEF}0.106 & \cellcolor[HTML]{EFEFEF}0.111 & \cellcolor[HTML]{EFEFEF}0.116 & \cellcolor[HTML]{EFEFEF}0.120 & \cellcolor[HTML]{EFEFEF}0.148 & \cellcolor[HTML]{EFEFEF}0.164 & \cellcolor[HTML]{EFEFEF}0.174 & \cellcolor[HTML]{EFEFEF}0.181 & \cellcolor[HTML]{EFEFEF}0.187 & \cellcolor[HTML]{EFEFEF}0.191 & \cellcolor[HTML]{EFEFEF}0.193 & \cellcolor[HTML]{EFEFEF}0.192 \\
 &  & 0.067 & 0.068 & 0.069 & 0.069 & 0.069 & 0.069 & 0.069 & 0.070 & 0.069 & 0.069 & 0.069 & 0.071 & 0.073 & 0.077 & 0.081 & 0.086 & 0.089 & 0.093 \\
 & \multirow{-2}{*}{$\textsc{HOPE}$} & \cellcolor[HTML]{EFEFEF}0.019 & \cellcolor[HTML]{EFEFEF}0.019 & \cellcolor[HTML]{EFEFEF}0.019 & \cellcolor[HTML]{EFEFEF}0.019 & \cellcolor[HTML]{EFEFEF}0.019 & \cellcolor[HTML]{EFEFEF}0.019 & \cellcolor[HTML]{EFEFEF}0.019 & \cellcolor[HTML]{EFEFEF}0.019 & \cellcolor[HTML]{EFEFEF}0.019 & \cellcolor[HTML]{EFEFEF}0.018 & \cellcolor[HTML]{EFEFEF}0.019 & \cellcolor[HTML]{EFEFEF}0.020 & \cellcolor[HTML]{EFEFEF}0.022 & \cellcolor[HTML]{EFEFEF}0.025 & \cellcolor[HTML]{EFEFEF}0.027 & \cellcolor[HTML]{EFEFEF}0.030 & \cellcolor[HTML]{EFEFEF}0.032 & \cellcolor[HTML]{EFEFEF}0.033 \\
  &  & {0.089} & {0.110} & {0.126} & {0.132} & {0.139} & {0.150} & {0.153} & {0.156} & {0.162} & {0.158} & {0.185} & {0.195} & {0.212} & {0.216} & {0.224} & {0.228} & {0.234} & {0.246} \\
 & \multirow{-2}{*}{$\textsc{VERSE}$} & \cellcolor[HTML]{EFEFEF}{0.052} & \cellcolor[HTML]{EFEFEF}{0.072} & \cellcolor[HTML]{EFEFEF}{\ul 0.090} & \cellcolor[HTML]{EFEFEF}{\ul 0.099} & \cellcolor[HTML]{EFEFEF}{\ul 0.107} & \cellcolor[HTML]{EFEFEF}{\textbf{0.116}} & \cellcolor[HTML]{EFEFEF}{\ul 0.120} & \cellcolor[HTML]{EFEFEF}{\ul 0.124} & \cellcolor[HTML]{EFEFEF}{\textbf{0.131}} & \cellcolor[HTML]{EFEFEF}{0.129} & \cellcolor[HTML]{EFEFEF}{0.157} & \cellcolor[HTML]{EFEFEF}{0.165} & \cellcolor[HTML]{EFEFEF}{0.183} & \cellcolor[HTML]{EFEFEF}{0.185} & \cellcolor[HTML]{EFEFEF}{0.193} & \cellcolor[HTML]{EFEFEF}{0.191} & \cellcolor[HTML]{EFEFEF}{0.200} & \cellcolor[HTML]{EFEFEF}{0.200} \\
 &  & {0.085} & {0.096} & {0.116} & {0.129} & {0.145} & {0.146} & {0.155} & {0.163} & {0.170} & {0.174} & {0.207} & {\ul 0.226} & {0.237} & {\ul 0.242} & {\textbf{0.248}} & {\textbf{0.256}} & {\textbf{0.258}} & {\textbf{0.257}} \\
 & \multirow{-2}{*}{$\textsc{ProNE}$} & \cellcolor[HTML]{EFEFEF}{0.039} & \cellcolor[HTML]{EFEFEF}{0.052} & \cellcolor[HTML]{EFEFEF}{0.065} & \cellcolor[HTML]{EFEFEF}{0.079} & \cellcolor[HTML]{EFEFEF}{0.091} & \cellcolor[HTML]{EFEFEF}{0.096} & \cellcolor[HTML]{EFEFEF}{0.102} & \cellcolor[HTML]{EFEFEF}{0.111} & \cellcolor[HTML]{EFEFEF}{0.117} & \cellcolor[HTML]{EFEFEF}{0.123} & \cellcolor[HTML]{EFEFEF}{0.155} & \cellcolor[HTML]{EFEFEF}{\textbf{0.180}} & \cellcolor[HTML]{EFEFEF}{\textbf{0.190}} & \cellcolor[HTML]{EFEFEF}{\textbf{0.196}} & \cellcolor[HTML]{EFEFEF}{\textbf{0.205}} & \cellcolor[HTML]{EFEFEF}{\textbf{0.216}} & \cellcolor[HTML]{EFEFEF}{\textbf{0.214}} & \cellcolor[HTML]{EFEFEF}{\textbf{0.206}}
\\
 &  & 0.073 & 0.084 & 0.091 & 0.097 & 0.104 & 0.112 & 0.116 & 0.120 & 0.124 & 0.129 & 0.155 & 0.170 & 0.180 & 0.187 & 0.192 & 0.195 & 0.198 & 0.199 \\
 & \multirow{-2}{*}{$\textsc{NetMF}$} & \cellcolor[HTML]{EFEFEF}0.031 & \cellcolor[HTML]{EFEFEF}0.043 & \cellcolor[HTML]{EFEFEF}0.052 & \cellcolor[HTML]{EFEFEF}0.059 & \cellcolor[HTML]{EFEFEF}0.066 & \cellcolor[HTML]{EFEFEF}0.073 & \cellcolor[HTML]{EFEFEF}0.078 & \cellcolor[HTML]{EFEFEF}0.083 & \cellcolor[HTML]{EFEFEF}0.086 & \cellcolor[HTML]{EFEFEF}0.091 & \cellcolor[HTML]{EFEFEF}0.116 & \cellcolor[HTML]{EFEFEF}0.129 & \cellcolor[HTML]{EFEFEF}0.137 & \cellcolor[HTML]{EFEFEF}0.143 & \cellcolor[HTML]{EFEFEF}0.147 & \cellcolor[HTML]{EFEFEF}0.150 & \cellcolor[HTML]{EFEFEF}0.152 & \cellcolor[HTML]{EFEFEF}0.150\\
 &  & 0.074 & 0.078 & 0.081 & 0.082 & 0.082 & 0.085 & 0.087 & 0.087 & 0.087 & 0.089 & 0.101 & 0.112 & 0.122 & 0.131 & 0.138 & 0.143 & 0.150 & 0.151 \\
 & \multirow{-2}{*}{$\textsc{GEMSEC}$} & \cellcolor[HTML]{EFEFEF}0.050 & \cellcolor[HTML]{EFEFEF}0.059 & \cellcolor[HTML]{EFEFEF}0.062 & \cellcolor[HTML]{EFEFEF}0.063 & \cellcolor[HTML]{EFEFEF}0.064 & \cellcolor[HTML]{EFEFEF}0.067 & \cellcolor[HTML]{EFEFEF}0.069 & \cellcolor[HTML]{EFEFEF}0.070 & \cellcolor[HTML]{EFEFEF}0.071 & \cellcolor[HTML]{EFEFEF}0.073 & \cellcolor[HTML]{EFEFEF}0.087 & \cellcolor[HTML]{EFEFEF}0.098 & \cellcolor[HTML]{EFEFEF}0.106 & \cellcolor[HTML]{EFEFEF}0.113 & \cellcolor[HTML]{EFEFEF}0.120 & \cellcolor[HTML]{EFEFEF}0.122 & \cellcolor[HTML]{EFEFEF}0.127 & \cellcolor[HTML]{EFEFEF}0.125\\
 &  & 0.085 & 0.097 & 0.103 & 0.112 & 0.115 & 0.119 & 0.122 & 0.123 & 0.126 & 0.128 & 0.137 & 0.143 & 0.148 & 0.153 & 0.155 & 0.162 & 0.165 & 0.168 \\
 \multirow{-18}{*}{\rotatebox{90}{Baselines}} & \multirow{-2}{*}{$\textsc{M-NMF}$} & \cellcolor[HTML]{EFEFEF}{\ul 0.057} & \cellcolor[HTML]{EFEFEF}0.071 & \cellcolor[HTML]{EFEFEF}0.080 & \cellcolor[HTML]{EFEFEF}0.089 & \cellcolor[HTML]{EFEFEF}0.093 & \cellcolor[HTML]{EFEFEF}0.097 & \cellcolor[HTML]{EFEFEF}0.102 & \cellcolor[HTML]{EFEFEF}0.103 & \cellcolor[HTML]{EFEFEF}0.104 & \cellcolor[HTML]{EFEFEF}0.108 & \cellcolor[HTML]{EFEFEF}0.119 & \cellcolor[HTML]{EFEFEF}0.125 & \cellcolor[HTML]{EFEFEF}0.131 & \cellcolor[HTML]{EFEFEF}0.135 & \cellcolor[HTML]{EFEFEF}0.136 & \cellcolor[HTML]{EFEFEF}0.141 & \cellcolor[HTML]{EFEFEF}0.142 & \cellcolor[HTML]{EFEFEF}0.141
\\\midrule
 &  & 0.087 & 0.102 & 0.112 & 0.121 & 0.128 & 0.134 & 0.138 & 0.142 & 0.145 & 0.148 & 0.167 & 0.180 & 0.186 & 0.192 & 0.196 & 0.200 & 0.202 & 0.203 \\
 & \multirow{-2}{*}{$\textsc{Gauss}$} & \cellcolor[HTML]{EFEFEF}0.038 & \cellcolor[HTML]{EFEFEF}0.051 & \cellcolor[HTML]{EFEFEF}0.060 & \cellcolor[HTML]{EFEFEF}0.067 & \cellcolor[HTML]{EFEFEF}0.073 & \cellcolor[HTML]{EFEFEF}0.078 & \cellcolor[HTML]{EFEFEF}0.083 & \cellcolor[HTML]{EFEFEF}0.086 & \cellcolor[HTML]{EFEFEF}0.089 & \cellcolor[HTML]{EFEFEF}0.092 & \cellcolor[HTML]{EFEFEF}0.112 & \cellcolor[HTML]{EFEFEF}0.125 & \cellcolor[HTML]{EFEFEF}0.132 & \cellcolor[HTML]{EFEFEF}0.138 & \cellcolor[HTML]{EFEFEF}0.142 & \cellcolor[HTML]{EFEFEF}0.146 & \cellcolor[HTML]{EFEFEF}0.148 & \cellcolor[HTML]{EFEFEF}0.147 \\
 &  & 0.103 & 0.126 & 0.142 & 0.154 & 0.164 & 0.172 & 0.179 & {\ul 0.185} & {\ul 0.190} & \textbf{0.195} & \textbf{0.220} & \textbf{0.232} & \textbf{0.239} & \textbf{0.244} & \textbf{0.248} & {\ul 0.250} & {\ul 0.254} & {\ul 0.256} \\
\multirow{-4}{*}{\rotatebox{90}{\scriptsize \textsc{KernelNE}}} & \multirow{-2}{*}{$\textsc{Sch}$} & \cellcolor[HTML]{EFEFEF}0.050 & \cellcolor[HTML]{EFEFEF}0.069 & \cellcolor[HTML]{EFEFEF}0.081 & \cellcolor[HTML]{EFEFEF}0.091 & \cellcolor[HTML]{EFEFEF}0.100 & \cellcolor[HTML]{EFEFEF}0.107 & \cellcolor[HTML]{EFEFEF}0.114 & \cellcolor[HTML]{EFEFEF}0.119 & \cellcolor[HTML]{EFEFEF}0.124 & \cellcolor[HTML]{EFEFEF}0.128 & \cellcolor[HTML]{EFEFEF}{\ul 0.156} & \cellcolor[HTML]{EFEFEF}{0.171} & \cellcolor[HTML]{EFEFEF}\textbf{0.181} & \cellcolor[HTML]{EFEFEF}{\ul 0.187} & \cellcolor[HTML]{EFEFEF}{\ul 0.193} & \cellcolor[HTML]{EFEFEF}{\ul 0.196} & \cellcolor[HTML]{EFEFEF}{\ul 0.200} & \cellcolor[HTML]{EFEFEF}{\ul 0.198} \\\midrule
 &  & {\ul 0.104} & \textbf{0.128} & {\ul 0.144} & \textbf{0.156} & {\ul 0.165} & \textbf{0.174} & \textbf{0.181} & \textbf{0.186} & \textbf{0.191} & \textbf{0.195} & \textbf{0.220} & {\ul 0.231} & {\ul 0.238} & {\ul 0.242} & {\ul 0.246} & {0.249} & {\ul 0.251} & {\ul 0.254} \\
 & \multirow{-2}{*}{$\textsc{Gauss}$} & \cellcolor[HTML]{EFEFEF}{\ul 0.053} & \cellcolor[HTML]{EFEFEF}0.071 & \cellcolor[HTML]{EFEFEF}0.083 & \cellcolor[HTML]{EFEFEF}0.094 & \cellcolor[HTML]{EFEFEF}0.104 & \cellcolor[HTML]{EFEFEF}0.109 & \cellcolor[HTML]{EFEFEF}0.117 & \cellcolor[HTML]{EFEFEF}0.122 & \cellcolor[HTML]{EFEFEF}{0.126} & \cellcolor[HTML]{EFEFEF}{\ul 0.131} & \cellcolor[HTML]{EFEFEF}\textbf{0.158} & \cellcolor[HTML]{EFEFEF}{\ul 0.172} & \cellcolor[HTML]{EFEFEF}{\ul 0.181} & \cellcolor[HTML]{EFEFEF}{\ul 0.187} & \cellcolor[HTML]{EFEFEF}{0.192} & \cellcolor[HTML]{EFEFEF}{0.195} & \cellcolor[HTML]{EFEFEF}{0.198} & \cellcolor[HTML]{EFEFEF}{0.197} \\
 &  & \textbf{0.105} & \textbf{0.128} & \textbf{0.145} & \textbf{0.156} & \textbf{0.167} & {\ul 0.173} & {\ul 0.180} & \textbf{0.186} & \textbf{0.191} & {\ul 0.194} & {\ul 0.219} & 0.230 & 0.236 & 0.241 & 0.245 & 0.247 & 0.249 & 0.251 \\
\multirow{-4}{*}{\rotatebox{90}{\scriptsize \textsc{MKernelNE}}} & \multirow{-2}{*}{$\textsc{Sch}$} & \cellcolor[HTML]{EFEFEF}{\ul 0.053} & \cellcolor[HTML]{EFEFEF}0.071 & \cellcolor[HTML]{EFEFEF}0.085 & \cellcolor[HTML]{EFEFEF}0.094 & \cellcolor[HTML]{EFEFEF}0.104 & \cellcolor[HTML]{EFEFEF}0.110 & \cellcolor[HTML]{EFEFEF}0.115 & \cellcolor[HTML]{EFEFEF}0.122 & \cellcolor[HTML]{EFEFEF}{0.126} & \cellcolor[HTML]{EFEFEF}0.130 & \cellcolor[HTML]{EFEFEF}0.155 & \cellcolor[HTML]{EFEFEF}0.169 & \cellcolor[HTML]{EFEFEF}{0.177} & \cellcolor[HTML]{EFEFEF}{0.184} & \cellcolor[HTML]{EFEFEF}0.188 & \cellcolor[HTML]{EFEFEF}0.192 & \cellcolor[HTML]{EFEFEF}0.193 & \cellcolor[HTML]{EFEFEF}0.193\\\\\bottomrule
\end{tabular}%
}
\end{table*}

\subsection{Datasets}
 In our experiments, we use eight networks of different types. To be consistent, we consider all network as undirected in all experiments, and the detailed statistics of the datasets are provided in Table \ref{tab:networks}. (i) \textsl{CiteSeer} \cite{harp} is a citation network obtained from the \textit{CiteSeer} library. Each node of the graph corresponds to a paper, while the edges indicate reference relationships among papers. The node labels represent the subjects of the paper. (ii) \textsl{Cora} \cite{cora} is another citation network constructed from the publications in the machine learning area; the documents are classified into seven categories. (iii) \textsl{DBLP} \cite{dblpdataset} is a co-authorship graph, where an edge exists between nodes if two authors have co-authored at least one paper. The labels represent the research areas. (iv) \textsl{PPI} (\textsl{Homo Sapiens}) \cite{node2vec} is a protein-protein interaction network for Homo Sapiens, in which biological states are used as node labels. (v) \textsl{AstroPh} \cite{astroph} is a collaboration network constructed from papers submitted to the  \textit{ArXiv} repository for the Astro Physics subject area, from January 1993 to April 2003. (vi) \textsl{HepTh} \cite{astroph} network is constructed in a similar way from the papers submitted to \textit{ArXiv} for the \textit{High Energy Physics - Theory} category. (vii) \textsl{Facebook} \cite{facebook} is a social network extracted from a survey conducted via a \textit{Facebook} application. (viii) \textsl{Gnutella} \cite{gnutella} is the peer-to-peer file-sharing network constructed from the snapshot collected in  August 2002 in which nodes and edges correspond to hosts and connections among them. 

\subsection{Parameter Settings}\label{subsec: paramsetting}
 For random walk-based approaches, we set the window size ($\gamma$) to $10$, the number of walks ($N$) to $80$, and the walk length ($L$) to $10$. We use the embedding size of $d=128$ for each method. The instances of \textsc{KernelNE} and \textsc{MKernelNE} are fed with random walks generated by \textsc{Node2Vec}.  For the training process, we adopt the negative sampling strategy \cite{word2vec} as described in Subsection \ref{subsec:single}. The negative samples are generated proportionally to its frequency raised to the power of $0.75$. For our methods and the baseline models needing to generate negative node instances, we consistently sample $5$ negative nodes for a fair comparison. In our experiments, the initial learning rate of stochastic gradient descent is set to $0.025$; then it is decreased linearly according to the number of processed nodes until the minimum value, $10^{-4}$. For the kernel parameters, the value of $\sigma$ has been chosen as $2.0$ for the single kernel version of the model (\textsc{KernelNE}). For \textsc{MKernelNE}, we considered three kernels and their parameters are set to $1.0, 2.0, 3.0$ and $1.0, 1.5, 2.0$ for \textsc{MkernelNE-Gauss} and \textsc{MkernelNE-Sch}, respectively. As an exception for the \textsl{PPI} network, we employed $0.5, 1.0, 1.5$ for \textsc{MKernelNE-Gauss} since it shows better performance due to the network's structure. The regularization parameters are always set to $\lambda=10^{-2}$ and $\beta=0.1$.

\subsection{Node Classification}
\noindent \textbf{Experimental setup.} In the node classification task, we have access to the labels of a certain fraction of nodes in the network (training set), and our goal is to predict the labels of the remaining nodes (test set). After learning the representation vectors for each node, we split them into varying sizes of training and test sets, ranging from $1\%$ up to $90\%$. The experiments are carried out by applying an one-vs-rest logistic regression classifier with $L_2$ regularization \cite{scikit-learn}. We report the average performance of $50$ runs for each representation learning method.

\vspace{.2cm}
\noindent \textbf{Experimental results.} \Cref{tab:classification_citeseer,tab:classification_cora,tab:classification_dblp,tab:classification_ppi}  report the Micro-$F_1$ and Macro-$F_1$ scores of the classification task. With boldface and underline we indicate the best and second-best performing model, respectively. As can be observed the single and multiple kernel versions of the proposed methodology outperform the baseline models, showing different characteristics depending on the graph dataset. While the \textit{Gaussian} kernel comes into prominence on the \textsl{Citeseer} and \textsl{Dblp} networks, the \textit{Schoenberg} kernel shows good performance for the \textsl{PPI} network. We further observe that leveraging multiple kernels with \textsc{MKernelNE} often has  superior performance compared to the single kernel, especially for smaller training ratios, which corroborates the effectiveness of the data-driven multiple kernel approach. In only a few cases, \textsc{Node2Vec} shows comparable performances. The reason stems from the structure of the network and the used random walk sequences. Since the distribution of the node labels does not always follow the network structure, we observe such occasional cases due to the complexity of the datasets concerning the node classification task.

\begin{table*}[t]
\centering
\caption{Area Under Curve (AUC) scores for the link prediction task.}
\label{tab:edge_prediction}
\begin{tabular}{rccccccccccccc}\toprule
\multicolumn{1}{l}{\textbf{}} & \multicolumn{9}{c}{Baselines} & \multicolumn{2}{c}{\textsc{KernelNE}} & \multicolumn{2}{c}{\textsc{MKernelNE}} \\\cmidrule(lr){2-10}\cmidrule(lr){11-12}\cmidrule(lr){13-14}
\multicolumn{1}{l}{\textbf{}} & \multicolumn{1}{l}{\rotatebox{80}{$\textsc{DeepWalk}$}} & \multicolumn{1}{l}{\rotatebox{80}{$\textsc{Node2vec}$}} & \multicolumn{1}{l}{\rotatebox{80}{$\textsc{Line}$}} & \multicolumn{1}{l}{\rotatebox{80}{$\textsc{HOPE}$}} & \rotatebox{80}{\textsc{ProNE}} & \rotatebox{80}{\textsc{VERSE}} & \multicolumn{1}{l}{\rotatebox{80}{$\textsc{NetMF}$}} & \multicolumn{1}{l}{\rotatebox{80}{$\textsc{GEMSEC}$}} &
\multicolumn{1}{l}{\rotatebox{80}{$\textsc{M-NMF}$}} &
\multicolumn{1}{l}{\rotatebox{80}{$\textsc{Gauss}$}} & \multicolumn{1}{l}{\rotatebox{80}{$\textsc{Sch}$}} & \multicolumn{1}{l}{\rotatebox{80}{$\textsc{Gauss}$}} & \multicolumn{1}{l}{\rotatebox{80}{$\textsc{Sch}$}}\\ \midrule
\textsl{Citeseer} & 0.828 & 0.827 & 0.725 & 0.517 & 0.719 & 0.754 & 0.818 & 0.713 & 0.634 & 0.807 & \textbf{0.882} & 0.850 & {\ul 0.863} \\
\textsl{Cora} & 0.779 & 0.781 & 0.693 & 0.548 & 0.667 & 0.737 & 0.767  & 0.716 & 0.616 & 0.765 & \textbf{0.818} & 0.792 & {\ul 0.807} \\
\textsl{DBLP} & 0.944 & 0.944 & 0.930 & 0.591 & 0.931 & 0.917 & 0.889 & 0.843 & 0.587 & 0.949 & {\ul 0.957} & \textbf{0.958} & {\ul 0.957} \\
\textsl{PPI} & {\ul 0.860} & \textbf{0.861} & 0.731 & 0.827 & 0.784 & 0.570 & 0.748 & 0.771 & 0.806 & 0.749 & 0.796 & 0.803 & 0.784 \\
\textsl{AstroPh} & 0.961 & 0.961 & 0.961 & 0.703 & 0.965 & 0.930 & 0.825 & 0.697 & 0.676 & 0.915 & {\ul 0.970} & \textbf{0.978} & 0.969 \\
\textsl{HepTh} & 0.896 & 0.896 & 0.827 & 0.623 & 0.862 & 0.840 & 0.844 & 0.708 & 0.633 & 0.897 & \textbf{0.915} & {\ul 0.914} & {\ul 0.914} \\
\textsl{Facebook} & 0.984 & 0.983 & 0.954 & 0.836 & 0.982 & 0.981 & 0.975 & 0.696 & 0.690 & 0.984 & {\ul 0.988} & \textbf{0.989} & \textbf{0.989} \\
\textsl{Gnutella} & {\ul 0.679} & 0.694 & 0.623 & \textbf{0.723} & 0.592 & 0.415 & 0.646 & 0.501 & 0.709 & 0.594 & 0.667 & 0.647 & 0.663\\\bottomrule
\end{tabular}%
\end{table*}

\subsection{Link Prediction}
For the link prediction task, we remove half of the network's edges while retaining its connectivity, and we learn node embedding on the residual network. For networks consisting of disconnected components, we consider the giant component among them as our initial graph. The removed edges constitute the positive samples for the testing set;  the same number of node pairs that do not exist  in the original graph is sampled at random to form the negative samples. Then, the entries of the feature vector corresponding to each node pair $(u, v)$ in the test set are computed based on the element-wise operation $|\mathbf{A}_{(v,i)}-\mathbf{A}_{(u,i)}|^2$, for each coordinate axis $i$ of the embedding vectors $\mathbf{A}_{(u,:)}$ and $\mathbf{A}_{(v,:)}$.  In the experiments, we use logistic regression with $L_2$ regularization.

\vspace{.2cm}
\noindent \textbf{Experimental results.} Table \ref{tab:edge_prediction} shows the \textit{Area Under Curve} (AUC) scores for the link prediction task. As we can observe, both the \textsc{KernelNE} and \textsc{MKernelNE} models perform well across different datasets compared to the other baseline models. Comparing now at the AUC score of \textsc{KernelNE} and \textsc{MKernelNE}, we observe that both models show quite similar behavior. In the case of the single kernel model \textsc{KernelNE}, the Schoenberg kernel (\textsc{Sch}) performs significantly better than the Gaussian one. On the contrary, in the \textsc{MKernelNE} model both kernels achieve similar performance. These results are consistent to the behavior observed in the node classification task.


\subsection{Parameter Sensitivity}
In this subsection, we examine how the performance of the proposed models is affected by the choice of parameters.    

\begin{figure}[!t]
\centering
\includegraphics[width=0.99\columnwidth]{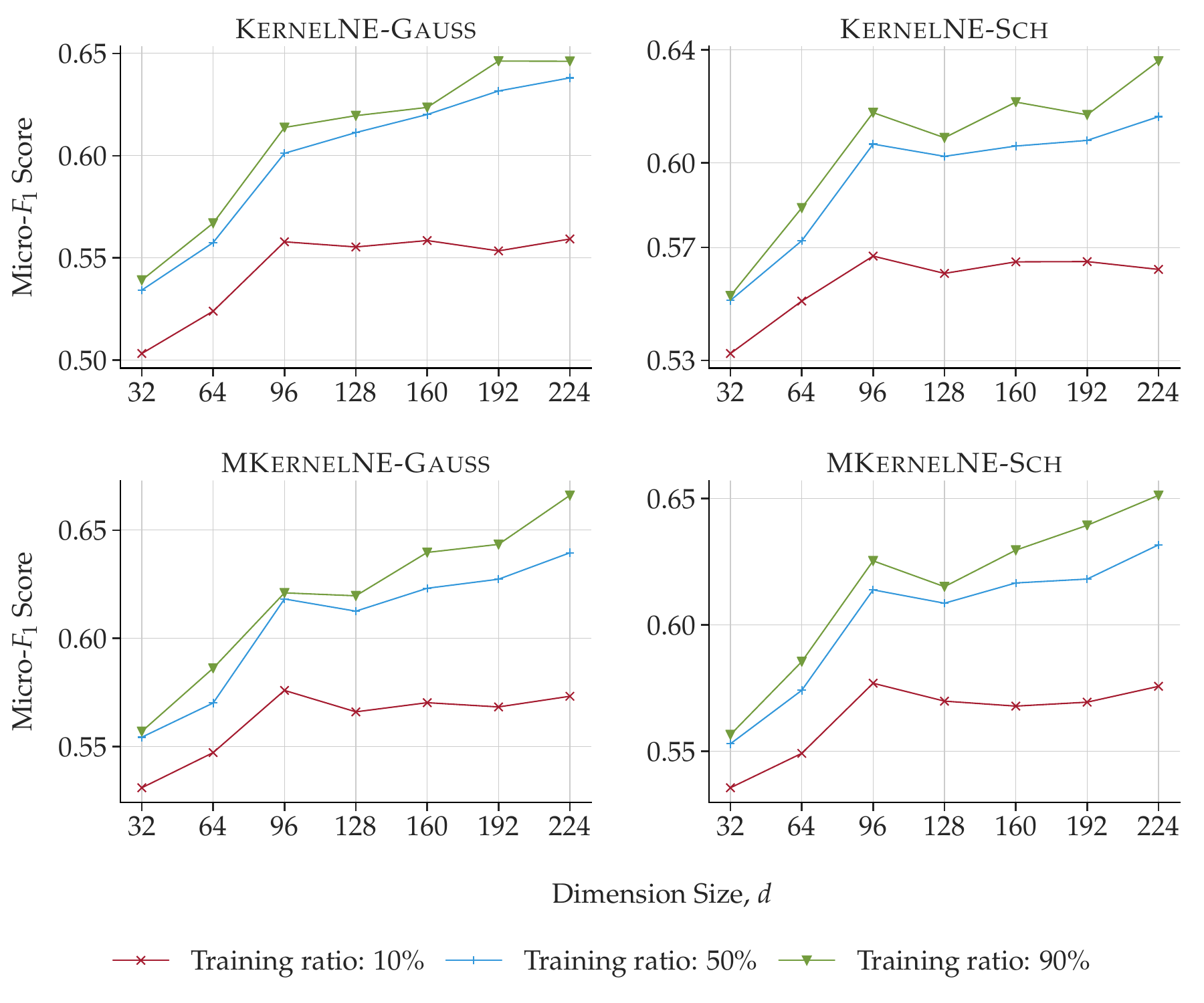}
\vspace{-.2cm}
\caption{Influence of the dimension size $d$ on the \textsl{CiteSeer} network.}
\label{fig:dimension_size}
\end{figure}

\vspace{.1cm}
\noindent \textbf{The effect of dimension size.} The dimension size $d$ is a critical parameter for node representation learning approaches since the desired properties of networks are aimed to be preserved in a lower-dimensional space. Figure \ref{fig:dimension_size} depicts the Micro-$F_1$ score of the proposed models for varying embedding dimension sizes, ranging from $d=32$ up to $d=224$  over the \textsl{Citeseer} network. As it can be seen,  all the different node embedding instances, both single and multiple kernel ones, have the same tendency with respect to $d$; the performance increases proportionally to the size of the embedding vectors. Focusing on 10\% training ratio, we observe that the performance gain almost stabilizes for embedding sizes greater than $96$.

\begin{figure}[!t]
\centering
\includegraphics[width=.99\columnwidth]{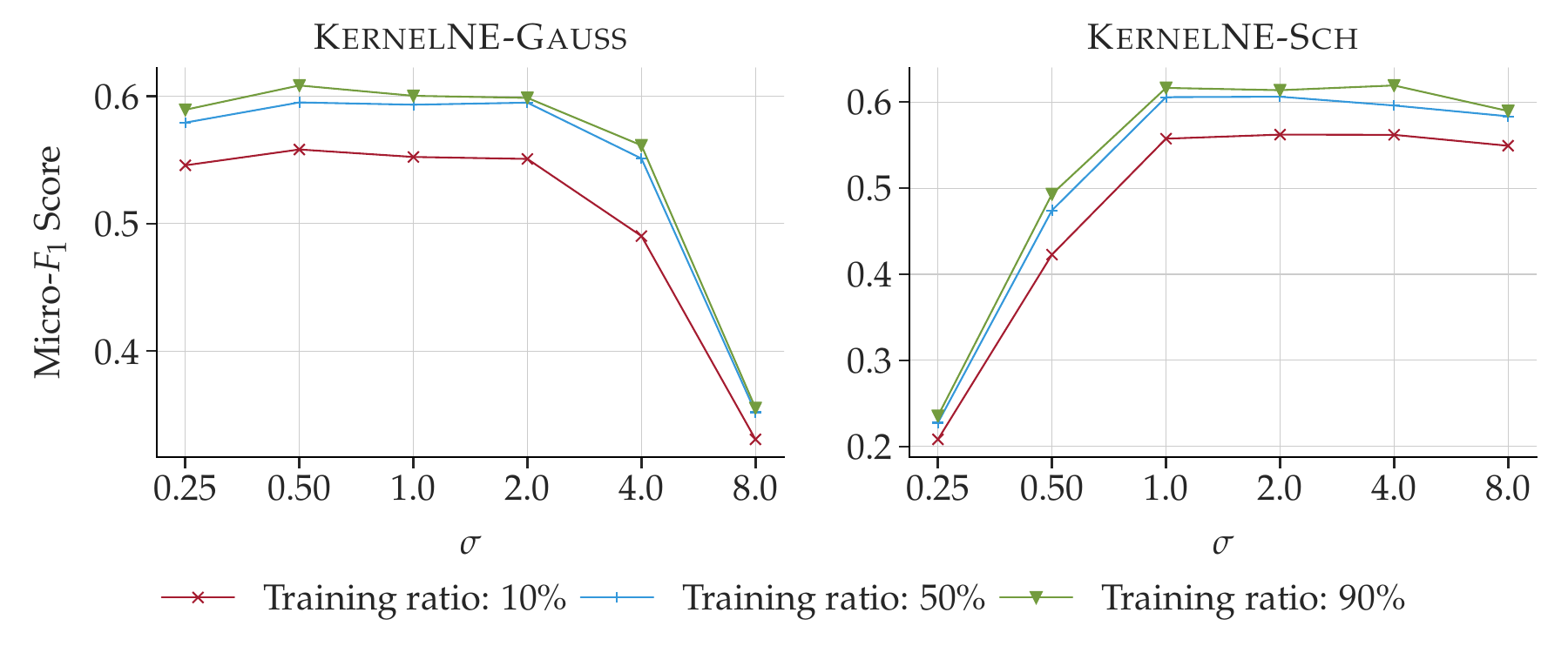}
\vspace{-.2cm}
\caption{Influence of kernel parameters on the \textsl{CiteSeer} network.}
\label{fig:kernel_params}
\end{figure}

\vspace{.1cm}
\noindent \textbf{The effect of kernel parameters.} We have further studied the behavior of the kernel parameter $\sigma$ of the Gaussian and Schoenberg kernel respectively (as described in Sec. \ref{subsec:single}). Figure \ref{fig:kernel_params} shows how the Micro $F_1$ node classification score is affected with respect to $\sigma$ for various training ratios on the \textsl{Citeseer} network. For the Gaussian kernel, we observe that the performance is almost stable for $\sigma$ values varying from $0.25$ up to $2.0$, showing a sudden decrease after $4.0$.  For the case of the Schoenberg kernel, we observe an almost opposite behavior. Recall that, parameter $\sigma$ has different impact on these kernels (Sec. \ref{subsec:single}). In particular, the Micro-$F_1$ score increases for $\sigma$ values ranging from $0.25$ to $0.50$, while the performance almost stabilizes in the range of $4.0$ to $8.0$. 

\subsection{Running Time Comparison}

\begin{table*}[]
\caption{Comparison of running time (in seconds) on an Erd\"os-Renyi random graph and the real-world networks used in the experiments. The symbol ``-'' denotes that the corresponding method is unable to run due to excessive memory requirements or it takes more than one day to complete.}
\label{tab:speedup}
\centering
\begin{tabular}{rccccccccccc}\toprule
 & \rotatebox{80}{\textsc{DeepWalk}} & \rotatebox{80}{\textsc{Node2Vec}} & \rotatebox{80}{\textsc{LINE}} & \rotatebox{80}{\textsc{HOPE}} & \rotatebox{80}{\textsc{VERSE}} & \rotatebox{80}{\textsc{ProNE}} & \rotatebox{80}{\textsc{NetMF}} & \rotatebox{80}{\textsc{GEMSEC}} & \rotatebox{80}{\textsc{M-NMF}} & \rotatebox{80}{\textsc{KernelNE}} & \rotatebox{80}{\textsc{MKernelNE}} \\\midrule
\textsl{Citeseer} & 136 & 42 & 2,020 & 26 & 294 & 5 & 34 & 2,012 & 530 & 93 & 186 \\
\textsl{Cora} & 139 & 40 & 2,069 & 27 & 256 & 5 & 22 & 2,064 & 392 & 87 & 170 \\
\textsl{Dblp} & 1,761 & 335 & 2,717 & 216 & 5,582 & 15 & 324 & 113,432 & 35,988 & 1,079 & 1,942 \\
\textsl{PPI} & 316 & 248 & 3,332 & 45 & 558 & 4 & 14 & 4,320 & 871 & 133 & 256 \\
\textsl{AstroPh} & 1,274 & 318 & 3,152 & 147 & 3,855 & 13 & 135 & 46,768 & 16,875 & 555 & 1,055 \\
\textsl{HepTh} & 804 & 128 & 3,866 & 77 & 1,264 & 5 & 84 & 15,838 & 3,776 & 346 & 625 \\
\textsl{Facebook} & 217 & 75 & 2,411 & 32 & 525 & 4 & 10 & 3,353 & 924 & 114 & 226 \\
\textsl{Gnutella} & 890 & 127 & 3,850 & 82 & 1,130 & 5 & 69 & 11,907 & 3,523 & 330 & 575 \\
\textsl{Erd\"os-Renyi} & 9,525 & 1,684 & 3,449 & 1,163 & 25,341 & 54 & 1,089 & - & - & 3,845 & 6,486\\\bottomrule
\end{tabular}%
\end{table*}

For running time comparison of the different models, we considered all the networks employed in the experiments. Additionally, we have generated an \textit{Erd\"os-Renyi} $\mathcal{G}_{n,p}$ graph model \cite{network_intro}, by setting the number of nodes $n=10^5$ and the edge probability $p=10^{-4}$. Table \ref{tab:speedup} reports the running time (in seconds) for both the baseline and the proposed models.

For this particular experiment, we have run all the models on a machine of $64$GB memory with a single thread when it is possible. The symbol "-" signifies that either the corresponding method cannot run due to excessive memory requirements or that it requires more than one day to complete. As we can observe, both the proposed models have comparable running time---utilizing multiple kernels with \textsc{MKernelNE} improves performance on downstream tasks without heavily affecting efficiency. Besides, \textsc{KernelNE} runs faster than \textsc{LINE} in most cases, while \textsc{MKernelNE} with two kernels shows also a comparable performance. It is also important to note that although matrix factorization-based models such as \textsc{M-NMF} are well-performing in some tasks, they are not very efficient because of the excessive memory demands (please see Sec. \ref{sec:related-work} for more details). On the contrary, our kernelized matrix factorization models are more efficient and scalable by leveraging negative sampling. The single kernel variant also runs faster than the random walk method \textsc{DeepWalk}.

\begin{figure}
\centering
\includegraphics[width=0.7\columnwidth]{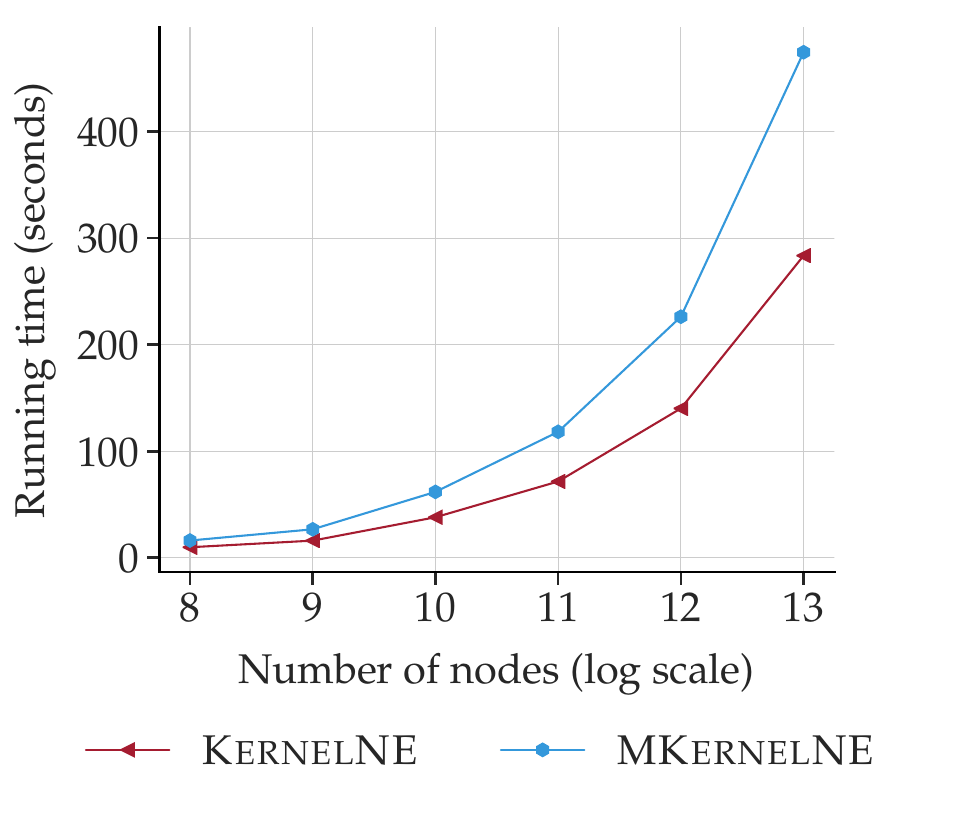}
\vspace{-.2cm}
\caption{Running time of \textsc{KernelNE} and \textsc{MKernelNE} models on Erd\"os-Renyi random graphs of different sizes.}
\label{fig:running_time}
\end{figure}

In addition to the running time comparison of the proposed approach against the baseline models, we further examine the running time of  \textsc{KernelNE} and \textsc{MKernelNE} (using three base kernels)  over Erd\"os-R\'enyi graphs of varying sizes, ranging from $2^8$ to $2^{13}$ nodes. In the generation of random graphs, we set the edge probabilities so that the expected node degree of graphs is $10$. Figure \ref{fig:running_time} shows the running time of the proposed models. Since the Gaussian and Schoenberg kernels have similar performance, we report the running time for the Gaussian kernel only. As we observe, considering multiple kernels does not significantly affect the scalability properties of the \textsc{MKernelNE} model.

\subsection{Visualization and Clustering of Embedding Vectors}

\begin{figure*}[t]
\centering
\includegraphics[width=0.85\textwidth]{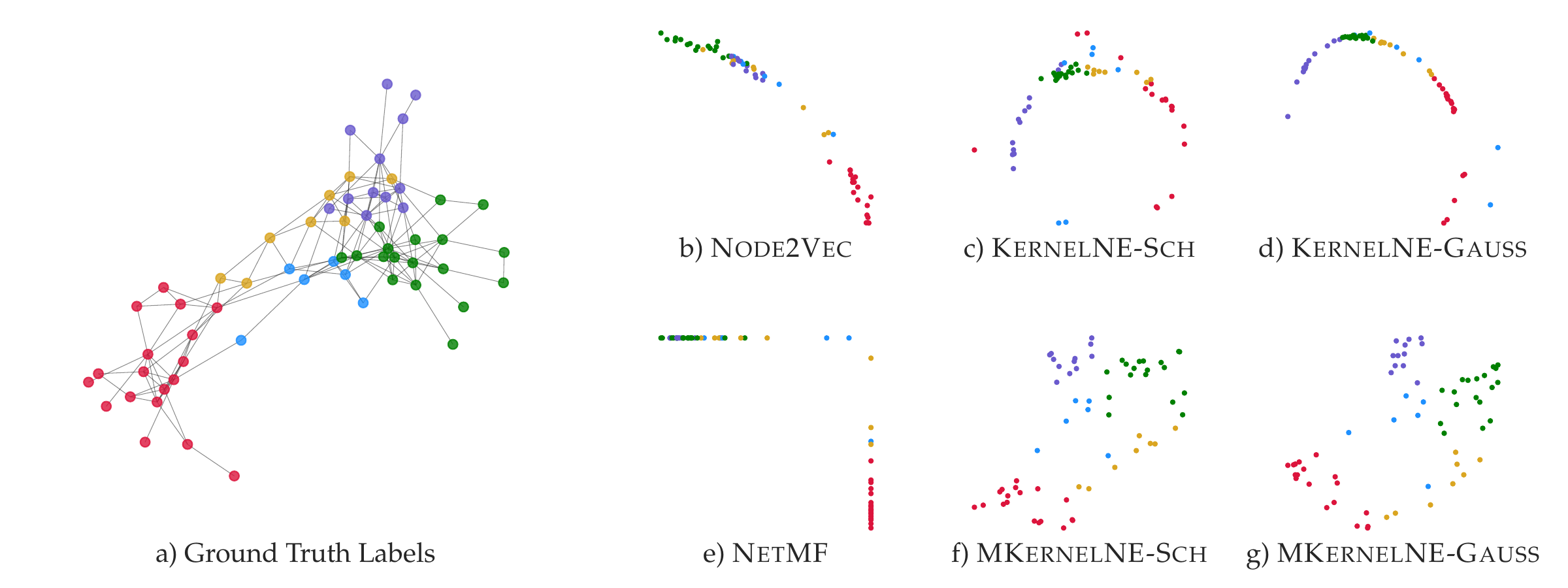}
\vspace{-.2cm}
\caption{The visualization of embeddings learned in 2D space for \textsl{Dolphins}. The colors indicate community labels computed by the \textsc{Louvain} algorithm.}
\label{fig:visualization_analysis}
\end{figure*}

\textit{Modularity} is a measure designed to assess the quality of the clustering structure of a graph  \cite{network_intro}. High modularity implies good clustering structure---the network consists of substructures in which nodes densely connected with each other. Here, we perform a simple visualization and clustering experiment to examine the ability of the different node embedding models to capture the underlying community structure, preserving network's modularity in the embedding space. Note that, to keep the settings of the experiment simple, we leverage the raw embedding vectors visualizing them in the two-dimensional space (instead of performing visualization with t-SNE \cite{t-SNE} or similar algorithms).

We perform experiments on the \textsl{Dolphins} \cite{lusseau2003bottlenose} toy network, which contains $62$ nodes and $159$ edges. We use the \textsc{Louvain} algorithm \cite{louvain} to detect the communities in the network. Each of these detected five communities is represented with a different color in Figure \ref{fig:visualization_analysis}a. We also use the proposed and the baseline models to learn node embeddings in two-dimensional space. 

As we can observe in Fig. \ref{fig:visualization_analysis}, different instances of \textsc{MKernelNE} learn embeddings in which nodes of the different communities are better distributed in the two-dimensional space. Besides, to further support this observation, we run the $k$-\textsc{means} clustering algorithm \cite{scikit-learn} on the embedding vectors, computing the corresponding \textit{normalized mutual information} (NMI) scores \cite{malliaros2013clustering}, assuming as ground-truth communities the ones computed by the \textsc{Louvain} algorithm in the graph space. While the NMI scores for \textsc{Node2Vec} and \textsc{NetMF} are $0.532$ and $0.572$ respectively, \textsc{KernelNE-Sch} achieves $0.511$ while for \textsc{KernelNE-Gauss}  we have $0.607$. The NMI scores significantly increase for proposed multiple kernel models \textsc{MKernelNE-Sch} and \textsc{MKernelNE-Gauss}, which are $0.684$ and $0.740$ respectively.
\section{Conclusion}
In this paper, we have studied the problem of learning node embeddings with kernel functions. We have first introduced \textsc{KernelNE}, a model that aims at interpreting random-walk based node proximity under a weighted matrix factorization framework, allowing to utilize kernel functions. To further boost performance, we have introduced \textsc{MKernelNE}, extending the proposed methodology to the multiple kernel learning framework. Besides, we have discussed how parameters of both models can be optimized via negative sampling in an efficient manner. Extensive experimental evaluation showed that the proposed kernelized models substantially outperform baseline NRL methods in node classification and link prediction tasks.

\par The proposed kernelized matrix factorization opens further research directions in network representation learning that we aim to explore in future work. To incorporate the sparsity property prominent in real-world networks, a probabilistic interpretation \cite{probabilisticMF-nips08} of the proposed matrix factorization mode would be suitable. Besides, it would be interesting to examine how the proposed models could be extended in the case of dynamic networks.


\vspace{.3cm}
\noindent \textbf{Acknowledgements.} Supported in part by ANR (French National Research Agency) under the JCJC project GraphIA (ANR-20-CE23-0009-01).

\appendices



\ifCLASSOPTIONcaptionsoff
  \newpage
\fi



%

\bibliographystyle{IEEEtran}
\bibliography{references}

%

\vspace{-1.0cm}
\begin{IEEEbiography}[{\includegraphics[width=1in,height=1.25in,clip,keepaspectratio]{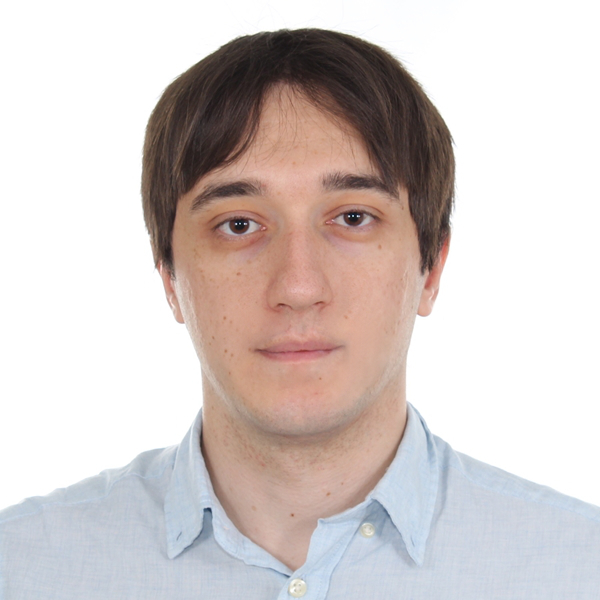}}]{Abdulkadir \c{C}elikkanat} is currently a postdoctoral researcher at the Section for Cognitive Systems of the Technical University of Denmark. He completed his Ph.D. at the Centre for Visual Computing of CentraleSupélec, Paris-Saclay University, and he was also a member of the OPIS team at Inria Saclay. Before his Ph.D. studies, he received his Bachelor degree in Mathematics and Master's degree in Computer Engineering from Bogazi\c{c}i University. His research mainly focuses on the analysis of graph-structured data. In particular, he is interested in graph representation learning and its applications for social and biological networks.
\end{IEEEbiography}
\vspace{-1.05cm}
\begin{IEEEbiography}[{\includegraphics[width=1in,height=1.25in,clip,keepaspectratio]{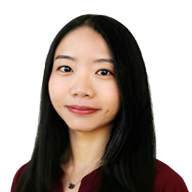}}]{Yanning Shen} is an assistant professor at University of California, Irvine. She
received her Ph.D. degree from the University of Minnesota (UMN) in 2019. She was a finalist for the Best Student Paper Award at the 2017 IEEE International Workshop on Computational Advances in Multi-Sensor Adaptive Processing, and the 2017 Asilomar Conference on Signals, Systems, and Computers. She was selected as a Rising Star in EECS by Stanford University in 2017, and received the Microsoft Academic Grant Award for AI Research in 2021. Her research interests span the areas of machine learning, network science, data science and signal processing.
\end{IEEEbiography}
\vspace{-1.05cm}
\begin{IEEEbiography}[{\includegraphics[width=1in,height=1.25in,clip,keepaspectratio]{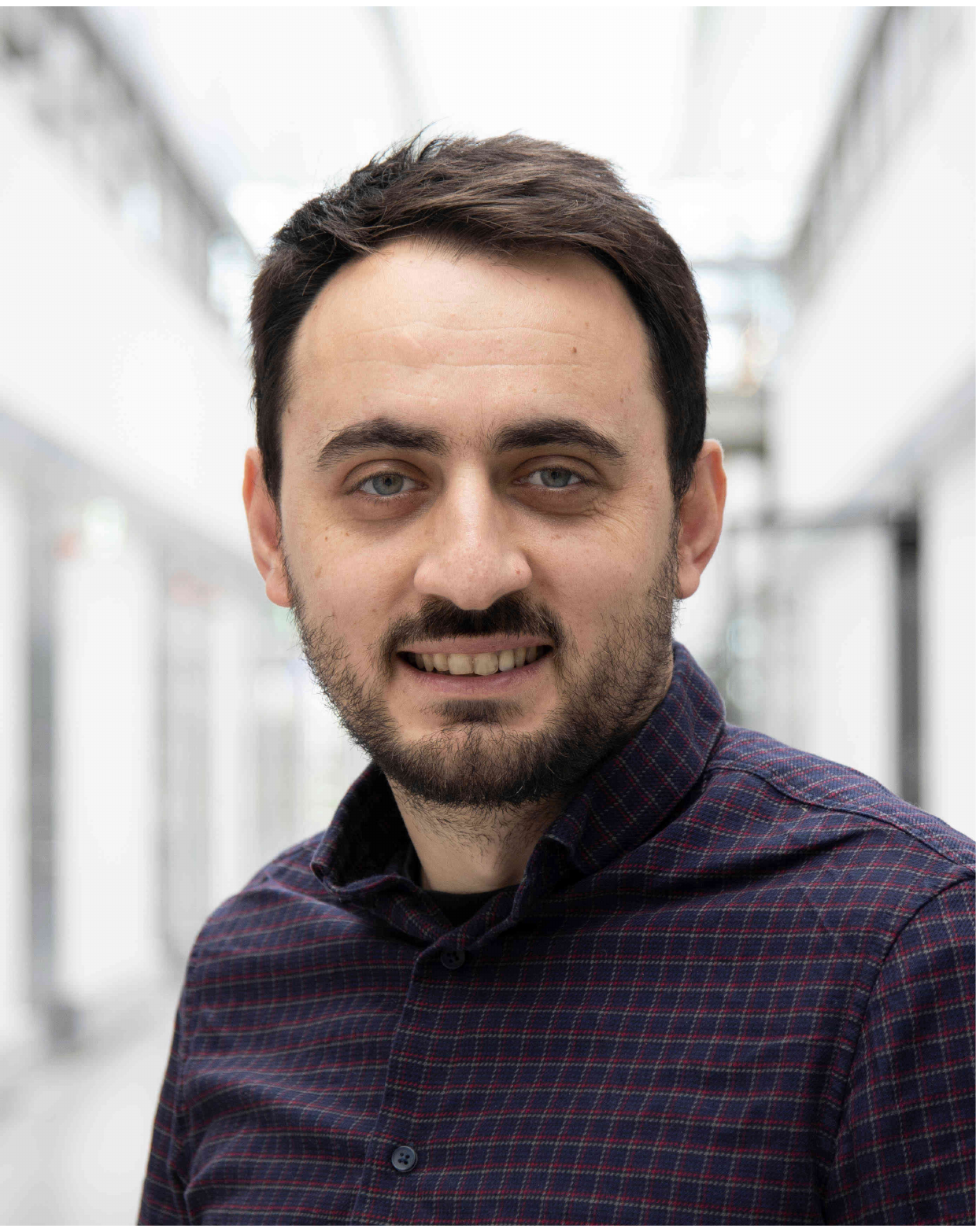}}]{Fragkiskos D. Malliaros}
 is an assistant professor at Paris-Saclay University, CentraleSupélec and associate researcher at
Inria Saclay. He also co-directs the Master Program in Data Sciences and Business Analytics (CentraleSupélec
and ESSEC Business School). Right before that, he was a postdoctoral researcher at UC San Diego (2016-17) and at
École Polytechnique (2015-16). He received his Ph.D. in Computer Science from École Polytechnique (2015) and his
M.Sc. degree from the University of Patras, Greece (2011). He is the recipient of the 2012 Google European Doctoral Fellowship in Graph Mining, the 2015 Thesis Prize by École Polytechnique, and best paper awards at TextGraphs-NAACL 2018 and ICWSM 2020 (honorable mention). In the past, he has been the co-chair of various data science-related workshops, and has also presented twelve invited tutorials at international conferences in the area of graph mining. His research interests span the broad area of data science, with focus on graph mining, machine learning, and network analysis.
\end{IEEEbiography}




\end{document}